\documentclass[12pt]{article}
\usepackage{amsmath,amssymb,amsthm, amsfonts} 
\usepackage[authoryear,round]{natbib}
\usepackage{hyperref}
\usepackage{graphicx, booktabs, multirow} 
\usepackage{setspace, enumitem}
\usepackage{float}
\usepackage{booktabs}
\usepackage{multirow}
\usepackage{caption}
\usepackage{subcaption}
\usepackage{cleveref}
\hypersetup{hidelinks, colorlinks=true, allcolors=black, pdfstartview=Fit, breaklinks=true}
\usepackage{lmodern, bm, mathrsfs, bbm, xcolor, soul}
\usepackage{appendix, authblk}
\numberwithin{equation}{section}
\allowdisplaybreaks[4]
\newtheorem{lemma}{Lemma}
\newtheorem{assumption}{Assumption}
\newtheorem{definition}{Definition}
\newtheorem{theorem}{Theorem}

\newcommand{\blind}{1}

\begin{document}
\date{} 
\def\spacingset#1{\renewcommand{\baselinestretch}{#1}\small\normalsize}\spacingset{1} 

\if1\blind
{
  \title{\fontsize{20pt}{30pt}\bf Deep Regression for Repeated Measurements under Covariate Shift}
  \author[a]{Yingxuan Wang}
  \author[a]{Xiangyu Xing} 
    \author[a] {Wangli Xu$^*$}
  \affil[a]{Center for Applied Statistics and School of Statistics, Renmin University of China,
Beijing, 100872, China}
  \maketitle
  \vspace{-1cm}
  \def\thefootnote{\fnsymbol{footnote}}
} \fi
\spacingset{1.9}
\if0\blind
{ 
  \bigskip\bigskip\bigskip
  \begin{center}
    {\fontsize{20pt}{30pt}\selectfont\bfseries Deep Regression for Repeated Measurements under Covariate-Shift}
  \end{center}
  \medskip
  \vspace{-0.8cm}
} \fi

\spacingset{1}
\bigskip
\begin{abstract}
This paper studies nonparametric regression with repeated measurements when the response in the target domain is unobservable or costly to collect. We adopt a transfer learning framework that leverages a source domain with observable responses under covariate shift. The target regression function is estimated by correcting the distribution shift via the density ratio. We consider both known and unknown density ratio scenarios, which reflect different data available for nonparametric regression estimation. 
In both cases, we further address two settings: the uniformly bounded density ratio and the unbounded case with finite moment conditions. 
Under the unknown density ratio scenario, both the density ratio and the target regression function are estimated using rectified linear unit (ReLU) feedforward neural networks (FNNs), whereas under the known density ratio scenario, only the target regression function is estimated by ReLU FNNs.
Theoretically, we establish non-asymptotic error bounds for the proposed estimators and prove that they achieve the minimax optimal convergence rate under the repeated measurements setting. 
Notably, we develop a novel approximation theory where the constants of the network parameters depend polynomially, rather than exponentially as in existing works, on the dimension, thereby mitigating the curse of dimensionality.
Consequently, we derive sharper non-asymptotic bounds for the stochastic error. The finite sample performance of the proposed method is demonstrated through numerical simulations and a real data application.
\end{abstract}


\spacingset{1.9}

\section{Introduction}
Repeated measurements, also called clustered data or longitudinal data,  have been widely encountered in statistical applications across diverse fields including biomedicine, economics and engineering \citep{frison1992repeated,iosif2014model,Yan02102025}.
For example, in PM2.5 concentrations monitoring data, each station represents a subject and concentration levels are measured at multiple time points.
Measurements are independent across different subjects yet correlated within the same subject, which is referred to as clustered dependence.
The correlation among repeated measurements may render   statistical methods for independent data inapplicable, leading to the development of many  methods for such data.
For instance, 
\citet{xu2016estimating} studied sufficient dimension reduction for repeated measures data. \citet{liu2026fixed} developed a Bayesian-motivated test for high-dimensional linear mixed models.

For repeated measurements data, extensive research has been conducted on modeling the impact of covariates on response variables, provided that both are observable. For the estimation of regression models, various nonparametric approaches have been proposed in the literature. 
For instance, 
\citet{Lin01062000} and \citet{2003Marginal} utilized kernel smoothing methods to investigate nonparametric estimation strategies for the correlation of clustered data.
\citet{hoover1998nonparametric}, \citet{Lv2020Nonparametric} and \citet{Deng2026Semiparametric} investigated the spline-based techniques to account for the relationship of covariate and response variables in longitudinal data.
Recently, due to the superior adaptivity and the ability to mitigate the curse of dimensionality in nonparametric estimation, deep neural networks (DNNs) have been widely researched, particularly the function class of feedforward neural networks (FNNs) with rectified linear unit (ReLU) 
\citep{yarotsky2017error,10.1214/18-AOS1747,jiao2023deep}.
Building on the favorable theoretical properties of DNNs, recent work by \citet{Yan02102025} adapted DNNs to repeated measurements data and established a comprehensive framework using ReLU FNNs, deriving minimax optimal bounds.

The aforementioned studies on repeated measurements data assume that both covariates and response variables are observed. However, in practical applications, most observations of the response variable in the dataset of interest are often unobserved, or costly or time-consuming to acquire, while the covariates are observed.
Studies such as \citet{2017Efficient}, \citet{oliver2019} and \citet{Azriel02102022} indicated that scenarios with unobserved responses are prevalent.
To study the effects of covariates on a response variable, one possible method is to employ transfer learning.
This approach treats the dataset of  interest 
as the target domain and leverages related datasets with observed responses as source domains
to perform statistical inference on the target domain. 
A fundamental assumption of transfer learning is that source and target domains share similar characteristics, with the differences often formalized as types of distribution shift, such as label shift, posterior shift, concept drift, and covariate shift. Under the scenario of independent data, transfer learning studies for cases where the response variable is rarely observed have already been conducted. Among others, 
\citet{pmlr-v80-lipton18a} proposed a black box shift estimation for label shift. 
\citet{cai2021transfer} developed a rate-optimal two-sample weighted K-nearest neighbors estimator under posterior drift.
\citet{Yang2022Concept} proposed a hybrid ensemble approach for concept drift-tolerant transfer learning problem.
\citet{feng2024deep} established nonparametric quantile regression with ReLU FNNs under covariate shift.

Among the various types of dataset shift, covariate shift is one of the most commonly studied scenarios. In this setting, the marginal distribution of covariates in the source domain differs from that in the target domain while the conditional distribution of the response given the covariates remains the same.
A critical approach to correcting covariate shift between the target and source domains is importance weighting via  the density ratio, as demonstrated in existing works like 
\citet{sugiyama2007direct,sugiyama2012density}, \citet{ma2023optimally} and \citet{feng2024deep}. 
However, in scenarios where the response variable is rarely observed, most transfer learning studies addressing covariate shift have been developed for independent data. For repeated measurements where both covariates and response variables are fully observed, existing research has shown that DNNs have superior estimation properties compared to other nonparametric methods, such as kernel or spline-based methods.
Motivated by the frequent absence of response variables in repeated measurements data and the favorable nonparametric properties of DNNs, this paper focuses on transfer learning under covariate shift and proposes a deep transfer learning framework. In this framework, both the density ratio and the target regression function are approximated by ReLU FNNs, enabling adaptive estimation with theoretical guarantees.
Our key contributions are summarized as follows:
\begin{itemize}
    \item [1.] We establish a theoretical framework for deep repeated measurements regression under covariate shift, adopting density ratio correction and considering two scenarios including when the density ratio is unknown and when it is known. 
    For each scenario, we derive non-asymptotic error upper bounds for the estimators under both uniformly bounded and unbounded density ratio settings. In the  unknown density ratio scenario, we split the source data into two disjoint subsets with one used for density ratio estimation and the other for target regression estimation. In the known scenario, we use the entire source data for the target regression estimation. Under the corresponding conditions in different density ratio settings, we prove that the resulting nonparametric estimators in both scenarios achieve the minimax optimal convergence rate, up to logarithmic factors.
    
    \item [2.] We show  that leveraging more source domain information in the known density ratio scenario leads to tighter error bounds for estimating the target regression function, a result that holds under both bounded and unbounded density ratio settings. We also demonstrate that, for the unknown density ratio scenario, optimal error bounds can be achieved by appropriately splitting the data to estimate the density ratio and the nonparametric regression function separately. We further show that under independent observations, i.e., one observation per subject, our error bound achieves the existing minimax rate, identifying it as a special case of our repeated measurements framework.

    \item [3.] We develop a novel  technique for bounding the approximation error where the constants of the network parameters depend polynomially on the covariate dimension,  rather than exponentially in contrast to existing works. Our approach mitigates the curse of dimensionality and allows for a smaller network size while achieving the same approximation rate, providing a corresponding theoretical guarantee. Moreover, we derive sharper non-asymptotic bounds for the stochastic error term, featuring a smaller poly-logarithmic exponent compared with existing literature. The effectiveness of our proposed method is demonstrated through extensive numerical analysis.

\end{itemize}

The remainder of the paper is organized as follows. Section \ref{Sec: Model Estimation} introduces the repeated measurements regression model under covariate shift and presents the estimators under different scenarios. 
This section also describes the function class of ReLU FNNs and the Hölder smoothness framework.
In Section \ref{sec: Theoretical Analysis}, we establish the non-asymptotic error bounds for different estimators and provide the corresponding theoretical guarantees.
Section \ref{sec: Numerical Studies} evaluates the proposed method via simulations and a real data example.
All technical proofs are included in Appendix.

We now introduce the notation used in the paper. We use $\mathbb{N}_{0}$, $\mathbb{N}$ and $\mathbb{R}$ to denote the set of nonnegative integers, positive integers and real numbers, respectively. For two real sequences $a_n$ and $b_n$, we define $a_n\lesssim b_n$ or equivalently $a_n=O(b_n)$ if there exists a constant $C>0$ such that $a_n\leq C{b}_{n}$ for all sufficiently large $n$. The notation $a_n=\Theta(b_n)$ indicates that $a_n\gtrsim b_n$, and $a_n\asymp b_n$ means that both $a_n \lesssim b_n$ and $a_n \gtrsim b_n$ hold.
  For a measurable function $f(\cdot):\Omega\to\mathbb{R}$ and $1\leq p<\infty$, the $L_{p}$ norm with respect to the distribution $P_{\boldsymbol{X}}$ is defined by $\|f\|_{L_{p}(P_{\boldsymbol{X}})}:=\left(\int_\Omega|f(\boldsymbol{X})|^pdP_{\boldsymbol{X}}\right)^{1/p}$ and the supremum norm is $\|f\|_\infty:=\sup_{\boldsymbol{X}\in\Omega}|f(\boldsymbol{X})|$. For a vector $\boldsymbol{\alpha}=(\alpha_1,\ldots,\alpha_d)^\top\in\mathbb{R}^d$, let $\|\boldsymbol{\alpha}\|_0:=\sum_{j=1}^d\mathbf{1}\{\alpha_j\neq0\}$, $\|\boldsymbol{\alpha}\|_1:=\sum_{j=1}^d|\alpha_j|$ and $\|\boldsymbol{\alpha}\|_\infty:=\max_{1\leq j\leq d}|\alpha_j|$, where $\mathbf{1}(\cdot)$ denotes the indicator function. 
The vectorization operator $\text{vec}(\boldsymbol{A})$ denotes the vector formed by stacking the columns of matrix $\boldsymbol{A}$ into a single column vector. For any $x \in \mathbb{R}$, let $\lceil x \rceil$ denotes the smallest integer greater than or equal to $x$.

\section{Model Estimation}
\label{Sec: Model Estimation}

In this paper, we consider repeated measurements data from both target and source domains. In the target domain, the response variable is unobserved, whereas it is fully observed in the source domains. Specifically, the target and source datasets are respectively denoted as $\mathcal{D}^{Q}=\{\boldsymbol{X}_{ij}^{Q}: 1\leq i\leq n^{Q}, 1\leq j\leq m_i^Q\}$ and $\mathcal{D}^{P}=\{(\boldsymbol{X}_{ij}^{P},Y_{ij}^{P}):1\leq i\leq n^{P}, 1\leq j\leq m_i^P\}$, where $\boldsymbol{X}_{ij}^{Q}$ represents the $j$-th covariate vector for subject $i$ in the target data, while $(\boldsymbol{X}_{ij}^{P},Y_{ij}^{P})$ denotes the $j$-th covariate vector and its corresponding response for subject $i$ in the source data. 

Assume that if the response variable $Y_{i j}^Q$ corresponding to $\boldsymbol{X}_{ij}^Q$ were observed, the data would follow the model
\begin{equation} \label{model_Q}
Y_{i j}^Q=f_0(\boldsymbol{X}_{i j}^Q)+f_i(\boldsymbol{X}_{i j}^Q)+\varepsilon_{i j},
\end{equation}
where $f_0(\boldsymbol{X}_{i j}^Q)=\mathbb{E}(Y_{ij}^Q \mid \boldsymbol{X}_{i j}^Q)$ is the regression function of interest, $f_{i}(\cdot)$ are the subject specific stochastic effect functions satisfying $\mathbb{E}(f_i(\boldsymbol{X}_{i j}^Q) \mid \boldsymbol{X}_{i j}^Q)=0$, and $\varepsilon_{i j}$ are independent noise variables. Our primary objective is to estimate $f_0(\cdot)$ in (\ref{model_Q}), and we denote the corresponding estimator by $\hat{f}_0(\cdot)$. If $Y_{i j}^Q$ were observed, a natural approach would be to minimize the following expression
\begin{equation} \label{min_Q}
    \hat{f}_0(\cdot) \in \arg\min_{f \in \mathcal{F}} \frac{1}{N^Q} \sum_{i=1}^{n_{Q}} \sum_{j=1}^{m_i^Q} (Y_{ij}^Q- f(\boldsymbol{X}_{ij}^Q))^2,
\end{equation}
where $N^Q$ is the total number of observations in the target domain, and $\mathcal{F}$ is a suitably chosen function class.

In fact, the response $Y_{i j}^Q$ is unobserved in the target domain. We therefore propose to estimate $\hat{f}_0(\cdot)$ via transfer learning under a covariate shift scenario. 
Let $P_{\boldsymbol{X}}$ and $Q_{\boldsymbol{X}}$ denote the marginal distributions of the covariates in the source and target domains, respectively, and let $P_{Y \mid \boldsymbol{X}}$ and $Q_{Y \mid \boldsymbol{X}}$ denote the corresponding conditional distributions of the response given the covariates. Under the covariate shift scenario, we assume that 
${P}_{\boldsymbol{X}}$ and ${Q}_{\boldsymbol{X}}$ can be different but ${P}_{Y|\boldsymbol{X}}={Q}_{Y|\boldsymbol{X}}$.
In this case, it can be derived that
\begin{equation} \label{den_rati}
\mathbb{E}_{(\boldsymbol{X},Y) \sim P_{\boldsymbol{X},Y}} [ r(\boldsymbol{X})  (Y - f_0(\boldsymbol{X}))^2 ]=\mathbb{E}_{(\boldsymbol{X}, Y) \sim Q_{\boldsymbol{X}, Y}} [(Y - f_0(\boldsymbol{X}))^2], 
\end{equation}
where $r(\boldsymbol{X})$ is the density ratio defined as $r(\boldsymbol{X})=q_{\boldsymbol{X}}(\boldsymbol{X})/p_{\boldsymbol{X}}(\boldsymbol{X})$, with  $p_{\boldsymbol{X}}(\cdot)$ and $q_{\boldsymbol{X}}(\cdot)$ denoting the probability density functions of $P_{\boldsymbol{X}}$ and $Q_{\boldsymbol{X}}$, respectively. Note that minimizing expression (\ref{min_Q}) essentially corresponds to the sample version of $\mathbb{E}_{(\boldsymbol{X}, Y) \sim Q_{\boldsymbol{X}, Y}} [(Y - f_0(\boldsymbol{X}))^2]$, according to  equation in (\ref{den_rati}), it can be replaced by the sample version of $\mathbb{E}_{(\boldsymbol{X},Y) \sim P_{\boldsymbol{X},Y}} [ r(\boldsymbol{X})\\(Y-f_0(\boldsymbol{X}))^2 ]$.
The estimation of $f_0(\cdot)$ is addressed under two distinct scenarios concerning the density ratio function $r(\cdot)$: (i) when it is unknown, and (ii) when it is known. We discuss these two cases separately below.

\begin{itemize}
    \item  \textbf{Scenario I: when  density ratio function $r(\cdot)$ is unknown.}
\end{itemize}

In this scenario, it is necessary to first estimate $r(\cdot)$ and subsequently estimate $f_0(\cdot)$ based on the obtained estimator of $r(\cdot)$. To establish theoretical properties for both estimators, the two functions are estimated using separate datasets. Specifically, 
we partition the source domain data $\mathcal{D}^{P}=\{(\boldsymbol{X}_{ij}^{P},Y_{ij}^{P}):1\leq i\leq n^{P}, 1\leq j\leq m_i^P\}$ into two disjoint subsets: $\mathcal{D}_1^{P}=\{(\boldsymbol{X}_{ij}^{P},Y_{ij}^{P}):1\leq i\leq n_1^{P}, 1\leq j\leq m_i^P\}$ and $\mathcal{D}_2^{P}=\{(\boldsymbol{X}_{ij}^{P},Y_{ij}^{P}):n_1^{P}+1 \leq i\leq n^{P}, 1\leq j\leq m_i^P\}$.
Following \citet{sugiyama2012density}, for any measurable function $v(\cdot)$, 
the density ratio $r(\cdot)$ is the minimizer of the following risk
\begin{equation} 
r(\cdot)=\arg\min_{v}L(v):= \arg\min_{v}\big \{\frac 1 2\,\mathbb{E}_{\boldsymbol{X}\sim P_{\boldsymbol{X}}}\!(v(\boldsymbol{X})^2)
    - \mathbb{E}_{\boldsymbol{X}\sim Q_{\boldsymbol{X}}}\!(v(\boldsymbol{X}))\big\}.
\label{r_least-squares}
\end{equation}
Therefore, an estimator $\hat r(\cdot)$ of the density ratio derived from $\mathcal{D}_2^{P}$ and $\mathcal{D}^{Q}$ can be obtained by
\begin{equation}
    \hat{r}(\cdot)\in\arg\min_{v \in \mathcal{V}}\widehat{L}(v):=
\arg\min_{v \in \mathcal{V}} \Big \{
\frac{1}{2 N_2^P}
\sum_{i=n_1^P+1}^{n^P}\sum_{j=1}^{m_i^{P}} v(\boldsymbol{X}_{ij}^P)^2
 - \frac{1}{N^Q}\sum_{i=1}^{n^Q}\sum_{j=1}^{m_i^{Q}} v(\boldsymbol{X}_{ij}^Q)\Big \}, \label{r_hat}
\end{equation}
where $\mathcal{V}$ is a suitably chosen function class, and $N_2^P$ is the total number of observations of dataset $\mathcal{D}_2^{P}$.
Therefore, the estimator $\hat{f}_{0,\hat{r}}(\cdot)$ for $f_0(\cdot)$,  constructed from dataset  $\mathcal{D}_1^{P}$,
can be obtained via the sample version of $\mathbb{E}_{(\boldsymbol{X},Y) \sim P_{\boldsymbol{X},Y}} [ r(\boldsymbol{X})  (Y - f_0(\boldsymbol{X}))^2 ]$, which is 
 \begin{equation}
\hat{f}_{0,\hat{r}}(\cdot)
\in
\arg\min_{f\in\mathcal{F}}
  \frac{1}{N_1^{P}}\sum_{i=1}^{n_1^{P}}\sum_{j=1}^{m^{P}_{i}}
   \hat{r}(\boldsymbol{X}_{ij}^{P})(Y_{ij}^{P} - f(\boldsymbol{X}_{ij}^{P}))^2,
\label{f_hat_r_hat}
\end{equation}
where, for notational simplicity, we use the same notation $\mathcal{F}$ for suitably chosen function classes throughout, although the specific class may vary, and $N_1^P$ is the total number of observations of dataset $\mathcal{D}_1^{P}$.

\begin{itemize}
    \item  \textbf{Scenario II:  when density ratio  function $r(\cdot)$ is known.}
\end{itemize}

In this scenario, estimating $r(\cdot)$ is not required. Instead of estimating $f_0(\cdot)$ using only the subset $\mathcal{D}_1^P$ of the source data, we may directly employ the full source dataset $\mathcal{D}^P$. Consequently, we obtain
 \begin{equation}
\hat{f}_{0,{r}}(\cdot)
\in
\arg\min_{f\in\mathcal{F}}
  \frac{1}{N^{P}}\sum_{i=1}^{n^{P}}\sum_{j=1}^{m^{P}_{i}}
   {r}(\boldsymbol{X}_{ij}^{P})(Y_{ij}^{P} - f(\boldsymbol{X}_{ij}^{P}))^2,
\label{f_hat_r_hat1}
\end{equation}
where $\mathcal{F}$ is a suitably chosen function class, and $N^P=N_1^P+N_2^P$ is the total number of observations in the source domain.

To estimate $\hat r(\cdot)$ in \eqref{r_hat},  $\hat{f}_{0,\hat{r}} (\cdot)$ in \eqref{f_hat_r_hat} and $\hat{f}_{0, {r}} (\cdot)$ in \eqref{f_hat_r_hat1},  we need to specify the function classes $\mathcal{V}$ and $\mathcal{F}$.
In this paper, we consider the function class of ReLU FNNs. 
We define a ReLU FNN as a function $\phi(\cdot)$ on $\mathbb{R}^d$ constructed recursively.
Let $\mathcal{D}$ denote the number of hidden layers, and let 
$N_l$ for $l=0, 1, \cdots, \mathcal{D}+1$
represent the widths of all layers, 
where $N_0 = d$ and $N_{\mathcal{D}+1} = 1$ are the input and output dimensions, respectively.
We define $\phi_0(\boldsymbol{X})=\boldsymbol{X}$ as the input layer and for $l=0,\dots,\mathcal{D}-1$, the $(l+1)$-th hidden layer is defined as $ \phi_{l+1}(\boldsymbol{x}) = \sigma(\boldsymbol{A}_l \phi_l(\boldsymbol{x}) + \boldsymbol{b}_l)$, 
where $\sigma(\cdot) = \max(0,\cdot)$ is the element-wise ReLU activation function 
and $\boldsymbol{A}_l \in \mathbb{R}^{N_{l+1}\times N_l}$ is the weight matrix and $\boldsymbol{b}_l \in \mathbb{R}^{N_{l+1}}$ is the bias vector. The final output of the network is given by
\begin{equation}
\phi(\boldsymbol{X})=\boldsymbol{A}_{\mathcal{D}}\phi_{\mathcal{D}}(\boldsymbol{X})+\boldsymbol{b}_{\mathcal{D}}, \quad 
\boldsymbol{A}_{\mathcal{D}}\in \mathbb{R}^{N_{\mathcal{D}+1}\times N_{\mathcal{D}}} ~\text{and} ~\boldsymbol{b}_{\mathcal{D}} \in \mathbb{R}^{N_{\mathcal{D}+1}},
\label{phi}
\end{equation} 
the set of all learnable parameters is denoted by $\boldsymbol{\theta} = \{(\boldsymbol{A}_l, \boldsymbol{b}_l)\}_{l=0}^{\mathcal{D}}$.
To quantify the complexity of the network, we denote that the width $\mathcal{W} = \max\left\{N_1, \dots, N_{\mathcal{D}}\right\}$ is the maximum width among all hidden layers, and the size $\mathcal{S} =\|\boldsymbol{\theta}\|_0 = \sum_{l=0}^{\mathcal{D}} \left(\|\text{vec}(\boldsymbol{A}_l)\|_0 + \|\boldsymbol{b}_l\|_0\right)$ is the number of non-zero elements in $\boldsymbol{\theta}$. The weight bound $\mathcal{B}$ is the maximum absolute value of all elements in $\boldsymbol{\theta}$, that is, $\|\boldsymbol{\theta}\|_\infty = \max_{l\in \{0, \dots, \mathcal{D}\}}
\max\{\|\text{vec}(\boldsymbol{A}_l)\|_\infty, \|\boldsymbol{b}_l\|_\infty\}\leq \mathcal{B}$. 

Based on these definitions, for the function $\phi(\cdot): \mathbb{R}^{d} \to \mathbb{R}$ defined in \eqref{phi}, the function class of ReLU FNNs with $\mathcal{W}$, $\mathcal{D}$, $\mathcal{S}$, and $\mathcal{B}$ is defined as
\begin{equation}
\label{FNN_d_1}
    \mathcal{F}_{d,1}(\mathcal{W},\mathcal{D},\mathcal{S},\mathcal{B}) = \bigl \{ \phi(\cdot): \phi(\cdot) \,  \text{satisfying } 
        \max_{l} N_l \leq \mathcal{W}, 
    \|\boldsymbol{\theta}\|_0 \leq \mathcal{S},  \|\boldsymbol{\theta}\|_\infty \leq \mathcal{B}
    \bigl\}.
\end{equation} 
The sparsity constraint $\mathcal{S}$ permits the weight matrices and bias vectors to contain zero entries. This structural sparsity effectively reduces the model complexity, allowing the derivation of sharper theoretical bounds.

To formalize the smoothness of the regression function $f(\cdot)$, we work with the Hölder class defined in the following definition.

\begin{definition}[H\"older Class $\mathcal{H}^\zeta(\text{[0,1]}^d, B)$]
For $\Omega = [0,1]^d$, the H\"older class is defined as:
\begin{equation} \label{holder}
\mathcal{H}^\zeta([0,1]^d, B) := \bigl\{ f: \Omega \to \mathbb{R} \,\big|\, \max_{\|\boldsymbol{\alpha}\|_1 \le t}\|\partial^{\boldsymbol{\alpha}} f\|_{\infty}  \leq B,\ \max_{\|\boldsymbol{\alpha}\|_1= t}\sup_{\substack{\boldsymbol{x},\boldsymbol{y}\in\Omega\\ \boldsymbol{x}\neq \boldsymbol{y}}} 
\frac{|\partial^{\boldsymbol{\alpha}} f(\boldsymbol{x}) - \partial^{\boldsymbol{\alpha}} f(\boldsymbol{y})|}{\|\boldsymbol{x} - \boldsymbol{y}\|_{\infty}^\sigma} \leq B \bigr\},
\end{equation}
where $t \in \mathbb{N}_0$, $\sigma \in (0,1]$, $\zeta = t + \sigma$, $\boldsymbol{\alpha} \in \mathbb{N}_0^d$ and \(B > 0\).
\end{definition}

Throughout the paper, we assume that the true regression function
$f_0(\cdot)$ belongs to $\mathcal{H}^\zeta([0,1]^d,B)$ for some $\zeta>0$
and $B>0$.

\section{Theoretical Analysis} \label{sec: Theoretical Analysis}
In this section, we establish non-asymptotic error upper bounds for the density ratio estimator $\hat{r}(\cdot)$, the estimators $\hat{f}_{0,\hat{r}}(\cdot)$, and $\hat{f}_{0,r}(\cdot)$ under different conditions.

\subsection{When density ratio function $r(\cdot)$ is unknown}
We investigate the theoretical properties of the estimators when the true density ratio $r(\cdot)$ is unknown. Specifically, we first establish a non-asymptotic bound for the mean squared error of the density ratio estimator $\hat{r}(\cdot)$ and then analyze final estimator $\hat{f}_{0,\hat{r}}(\cdot)$. In the following, we introduce some technical assumptions for theoretical analysis under different scenarios. For notational simplicity, we denote $q_{\boldsymbol{X}}(\boldsymbol{X})$ and $p_{\boldsymbol{X}}(\boldsymbol{X})$ as $q_{\boldsymbol{X}}$ and $p_{\boldsymbol{X}}$, respectively.

\begin{assumption}\label{assumption rsupp}
The support of the target density $q_{\boldsymbol{X}}$ is contained in that of the source density $p_{\boldsymbol{X}}$, denoted as $\mathrm{supp}(q_{\boldsymbol{X}}) \subseteq \mathrm{supp}(p_{\boldsymbol{X}})$.
\end{assumption}
\begin{assumption}
    \label{assumption rb} The density ratio $r(\cdot)$ is uniformly bounded, i.e., $\Gamma:=\operatorname*{sup}_{\boldsymbol{X}\in\mathcal{X}}r(\boldsymbol{X})<\infty.$
\end{assumption}
\begin{assumption}
    \label{assumption rholder}
The density ratio $r(\cdot)$ is H\"older continuous that $r(\cdot) \in \mathcal{H}^\alpha([0,1]^d, \Gamma)$ for some constant $\alpha>0$.
\end{assumption}
\begin{assumption}
    \label{assumption rb0}
    The density ratio $r(\cdot)$ has a strictly positive lower bound, that is, $M:=\inf_{\boldsymbol{X}\in\mathcal{X}}r(\boldsymbol{X})>0$.
\end{assumption}
\begin{assumption}
    \label{assumption rub} There exists $\delta>0$ such that the density ratio $r(\cdot)$ is unbounded but has a finite $(\delta+2)$-th moment with respect to $P_{\boldsymbol{X}}$, namely, $\kappa_{\delta}:=\mathbb{E}_{\boldsymbol{X}\sim P_{\boldsymbol{X}}}[r^{\delta+2}(\boldsymbol{X})]<\infty$.
\end{assumption}

Assumption \ref{assumption rsupp} 
requires that the support of the target probability density function is contained in that of the source, which is widely adopted in covariate shift works \citep{ma2023optimally,feng2023towards}  to ensure that the density ratio is well-defined.
Assumption \ref{assumption rb} requires that the density ratio is uniformly upper bounded by $\Gamma$. Assumption \ref{assumption rholder} 
imposes a standard smoothness condition on the density ratio function which plays a crucial role in establishing consistency for the estimator. Assumption \ref{assumption rb0} requires that the density ratio is uniformly bounded away from zero. 
Assumption \ref{assumption rub}
requires the density ratio is unbounded but has a finite $(\delta+2)$-th moment, relaxing the strict uniformly bounded condition.

In the following, to investigate the theoretical properties of $\hat{f}_{0,\hat{r}}(\cdot)$ in \eqref{f_hat_r_hat}, we begin by analyzing the density ratio estimator $\hat{r}(\cdot)$ under two scenarios where the true density ratio $r(\cdot)$ is uniformly bounded by $\Gamma>0$ or is unbounded but has a finite $(\delta+2)$-th moment.
We first establish a non-asymptotic error bound for $\hat{r}(\cdot)$ under  uniformly bounded case. 
Following  \citet{jiao2023deep}, 
we decompose the mean squared error $\mathbb{E}[\|\hat{r}-r\|_{L_2(P_{\boldsymbol{X}})}^2]$ into a stochastic error bound and an approximation error, i.e. $\mathbb{E}[\|\hat{r}-r\|_{L_2(P_{\boldsymbol{X}})}^2] \lesssim  \mathbb{E}[L(r)-2\widehat{L}(\hat{r})+L(\hat{r})]+ \inf_{v\in \mathcal{V}}\|r-v\|_{L_2(P_{\mathbf{X}})}^2$, where $L(\cdot)$ and $\hat L(\cdot)$ are defined in (\ref{r_least-squares}) and (\ref{r_hat}) respectively.
Theorem \ref{Theorem rstochastic} below presents the stochastic error bound.
\begin{theorem}
    \label{Theorem rstochastic}
    Suppose that Assumptions \ref{assumption rsupp}, \ref{assumption rb} and \ref{assumption rholder} are satisfied. Let $\mathcal{V}$ in $\widehat{L}(\cdot)$ be the function class in \eqref{FNN_d_1} with the depth $\mathcal{D}$ and the size $\mathcal{S}$, and assume that $\mathcal{V}$ is uniformly bounded by $\Gamma$. Then, the stochastic error term $\mathbb{E}[L(r)-2\widehat{L}(\hat{r})+L(\hat{r})]$ satisfies 
    \begin{center}
        $\mathbb{E}[L(r)-2\widehat{L}(\hat{r})+L(\hat{r})]\lesssim \Gamma(\Gamma+1)\mathcal{DS}\log\mathcal{S}\log N_r/N_r$,
    \end{center}
    where $1/N_r=1/N_2^P+1/N^Q$.
\end{theorem}

Using Bernstein inequality, Theorem \ref{Theorem rstochastic} provides a generic upper bound on the stochastic error in terms of the network depth $\mathcal{D}$, the network size $\mathcal{S}$ and the effective sample size $N_r$, for any candidate function class $\mathcal{V}$ in \eqref{FNN_d_1} that is uniformly bounded by $\Gamma$. To bound the approximation error term, we apply a general ReLU FNN approximation result for Hölder functions in Theorem \ref{Theorem approx}.

\begin{theorem}
    \label{Theorem approx}
    Suppose that $f(\cdot) \in \mathcal{H}^\zeta([0,1]^d, B)$. For any $\varepsilon \in (0,1)$, 
there exists a ReLU FNN $\phi(\cdot)\in \mathcal{F}_{d,1}$ defined in \eqref{FNN_d_1} with the depth $\mathcal{D}=O(\log(1/\varepsilon))$, size $\mathcal{S}=O(\varepsilon^{-d/\zeta}\log(1/\varepsilon))$ and weights bound $\mathcal{B}=O(\varepsilon^{-d/\zeta})$, such that
\begin{center}
    $\|f - \phi\|_{\infty} \leq \varepsilon$.
\end{center}
\end{theorem}

It is noteworthy that we employ a novel simplicial partition of unity, which involves only $d+1$ overlapping local regions at any given point, ensuring that the constants in the upper bound of the network parameters $\mathcal{D}$ and $\mathcal{S}$ depend on $d+1$, as shown in \eqref{eq:N,delte} of Appendix \ref{AppendixA.1.2}. 
In contrast, many existing works \citep{yarotsky2017error,10.1214/19-AOS1875,feng2024deep} utilize a partition of unity based on tensor products to investigate the approximation error, where the constants in the network parameters depend on $2^d$. 
Consequantly, our construction attains the same upper bound $\varepsilon$ with a smaller network size and simultaneously mitigates the curse of dimensionality. According to Theorem \ref{Theorem approx}, if $r(\cdot)\in \mathcal{H}^\zeta([0,1]^d, B)$ and $v(\cdot)\in \mathcal{F}_{d,1}$, the approximation error satisfies $\inf_{v\in \mathcal{V}}\|r-v\|_{L_2(P_{\mathbf{X}})}^2\leq \varepsilon^2$. Since the size satisfies  $\mathcal{S}=O(\varepsilon^{-d/\zeta}\log(1/\varepsilon))$, $\varepsilon^2$ yields a rate of order $O(\mathcal{S}^{-2\zeta/d})$ up to logarithmic factors. The following theorem establishes an upper bound for the mean squared error of $\hat{r}(\cdot)$

\begin{theorem}
    \label{Theorem 3rb}
    Suppose that Assumptions \ref{assumption rsupp}, \ref{assumption rb} and \ref{assumption rholder} are satisfied, and let $\mathcal{V}$ be the function class in \eqref{FNN_d_1} bounded by $\Gamma$ with the depth $\mathcal{D}=O(\log N_r)$, the weights $\mathcal{B}=O(N_r^{d/(2\alpha+d))})$ and the size $\mathcal{S}=O(N_r^{d/(2\alpha+d)}\log N_r)$, where 
    $1/N_r=1/N_2^P+1/N^Q$. 
     Then, the mean squared error of the estimator $\hat{r}(\cdot)$ in \eqref{r_hat} satisfies
\begin{center}
$\mathbb{E}[\|\hat{r}-r\|_{L_2(P_{\boldsymbol{X}})}^2]\lesssim \Gamma(\Gamma+1)N_r^{-{2\alpha}/(d+2\alpha)}(\log N_r)^5$.
\nonumber
\end{center}
\end{theorem}

Theorem \ref{Theorem 3rb} establishes a non-asymptotic error bound for the density ratio estimator $\hat{r}(\cdot)$ in the uniformly bounded case. The bound is obtained by plugging the stochastic and approximation bounds from the results of Theorems \ref{Theorem rstochastic} and \ref{Theorem approx} into the error decomposition. Specifically, under the network parameter scaling specified in Theorem \ref{Theorem 3rb}, we substitute the corresponding choices of $\mathcal{D}$ and $\mathcal{S}$ into Theorem \ref{Theorem rstochastic} to control the stochastic term and then we apply Theorem \ref{Theorem approx} with $\varepsilon=(N_r)^{-\alpha/(2\alpha+d)}$, which implies $\inf_{v\in \mathcal{V}}\|r-v\|_{L_2(P_{\mathbf{X}})}^2 \lesssim N_r^{-{2\alpha}/(d+2\alpha)}\log N_r$. Combining these two bounds yields a convergence rate of order $O((N_r)^{-2\alpha/(d+2\alpha)})$ up to logarithmic factors, which achieves the minimax optimal. We introduce the following two assumptions to establish the theoretical properties for the estimators of $f_0(\cdot)$. 

\begin{assumption}\label{assumption f0B1}
    The mean function satisfies $\|f_{0}\|_\infty \leq B_1$ for some positive constant $B_1$.
\end{assumption}
\begin{assumption}\label{assumption fiepsilon}
    The subject specific stochastic effect functions $f_{i}(\cdot)$ are continuous, and there exists nonnegative numbers $B_{2}$ and $B_{3}$, such that for $1\leq i\leq n^{P},1\leq j\leq m^{P}_{i}$,
    \begin{center}
    $\mathbb{E}[\exp(|f_{i}(\boldsymbol{X}_{ij}^{P})/B_{2}|)]\leq1,\quad\mathbb{E}[\exp(|\varepsilon_{i j}/B_{3}|)]\leq1$.
    \end{center}
\end{assumption}

Assumption \ref{assumption f0B1} is a general and mild boundedness condition. To satisfy this assumption, the candidate functions are often truncated. 
Assumption \ref{assumption fiepsilon} requires that the random components have sub-exponential tails, which is a common requirement for establishing concentration inequalities.
The following theorem establishes a non-asymptotic error bound for the estimator $\hat{f}_{0,\hat{r}}(\cdot)$ defined in \eqref{f_hat_r_hat}.

\begin{theorem}
    \label{Theorem Ef^_0_r^}
    Suppose that Assumptions \ref{assumption rsupp}
    -\ref{assumption rb0}, \ref{assumption f0B1} and \ref{assumption fiepsilon} are satisfied, and let $\mathcal{F}$ be the function class in \eqref{FNN_d_1} bounded by $B_1$ with the depth $\mathcal{D}=O(\log N_1^P)$, the weights $\mathcal{B}=O((N_1^P)^{d/(d+2\zeta)})$ and the size $\mathcal{S}=O((N_1^P)^{d/(2\zeta+d)}(\log N_1^P)^{-5d/(2\zeta+d)})$. Then, the mean squared error of the estimator $\hat{f}_{0,\hat{r}}(\cdot)$ defined in \eqref{f_hat_r_hat} with the density ratio estimator $\hat{r}(\cdot)$ in \eqref{r_hat} satisfies
\begin{align*}
\mathbb{E}[\|\hat{f}_{0,\hat{r}}-f_0\|_{L_2(Q_{\boldsymbol{X}})}^2]\lesssim &\Gamma((n_1^P)^{-1}+(N_1^P)^{-2\zeta/(2\zeta+d)}(\log N_1^P)^{10\zeta/(2\zeta+d)})\\
& +\Gamma(\Gamma+1)(N_r)^{-2\alpha/(d+2\alpha)}(\log N_r)^{5}/M.
\end{align*}
Moreover, if $N_r\gtrsim((\Gamma+1)/M)^{(d+2\alpha)/2\alpha}(N_1^P)^{\zeta(d+2\alpha)/\alpha(d+2\zeta)}$, then we have
\begin{center} 
$\mathbb{E}[\|\hat{f}_{0,\hat{r}}-f_0\|_{L_2(Q_{\boldsymbol{X}})}^2] \lesssim \Gamma((n_1^P)^{-1}+(N_1^P)^{-2\zeta/(2\zeta+d)}(\log N_1^P)^{10\zeta/(2\zeta+d)})$.
\end{center}
\end{theorem}

The result reflects the trade-off between the estimator $\hat{f}_{0,\hat{r}}(\cdot)$ and the density ratio estimator $\hat{r}(\cdot)$. In particular, the term $\Gamma((n_1^P)^{-1}+(N_1^P)^{-2\zeta/(2\zeta+d)}(\log N_1^P)^{10\zeta/(2\zeta+d)})$ captures the error incurred by estimating $f_0(\cdot)$, whereas $\Gamma(\Gamma+1)(N_r)^{-2\alpha/(d+2\alpha)}(\log N_r)^{5}/M$ accounts for the error resulted from estimating the density ratio $r(\cdot)$.
When $N_r$ is sufficiently large, specifically $N_r\gtrsim((\Gamma+1)/M)^{(d+2\alpha)/2\alpha}(N_1^P)^{\zeta(d+2\alpha)/\alpha(d+2\zeta)}$, the error bound simplifies to $\Gamma((n_1^P)^{-1}+(N_1^P)^{-2\zeta/(2\zeta+d)}(\log N_1^P)^{10\zeta/(2\zeta+d)})$.
Compared with \citet{Yan02102025}, our repeated measurement analysis yields a smaller poly-logarithmic exponent.
Consequently, under the above condition on $N_r$, our estimator achieves the minimax optimal rate of $O((n_1^P)^{-1}+(N_1^P)^{-2\zeta/(2\zeta+d)})$ up to logarithmic factors.

We now relax the uniformly bounded assumption on $r(\cdot)$ and extend our analysis to the setting where $r(\cdot)$ is unbounded but satisfies a finite $(\delta+2)$-th moment condition.
In this setting, the true $r(\cdot)$ may take extremely large values at certain points, causing the estimator to be highly sensitive to these points.
To control this instability and ensure theoretical robustness, we employ a truncation mechanism for the density ratio estimator.
For a fixed truncation level $\xi>0$, we define the truncated  density ratio estimator as
\begin{equation}
\hat{r}_\xi(\cdot)
\in \arg\min_{v \in T_\xi \mathcal{V}} \Big \{
\frac{1}{2 N_2^P}
\sum_{i=n_1^P+1}^{n^P}\sum_{j=1}^{m_i^{P}} v(\boldsymbol{X}_{ij}^P)^2
 - \frac{1}{N^Q}\sum_{i=1}^{n^Q}\sum_{j=1}^{m_i^{Q}} v(\boldsymbol{X}_{ij}^Q)\Big \}, 
\label{hat_r_xi} 
\end{equation}
where $T_\xi \mathcal{V} = \{T_\xi v: v \in \mathcal{V}\}$ denotes the truncated function class with the clipping operator $T_\xi v := \min\{v, \xi\}$. Note that $T_\xi\mathcal{V}$ can be naturally implemented by appending an activation layer $\sigma_\xi(\cdot)=\min\{\cdot,\xi\}$ to the network output, i.e., $T_\xi v := \sigma_\xi \circ v$. The following theorem establishes the mean squard error bound for $\hat{r}_\xi(\cdot)$.

\begin{theorem}
    \label{Theorem r_xi_hat}
    Suppose that Assumptions \ref{assumption rsupp} and \ref{assumption rub}
    are satisfied, and there exists some constant $\alpha>0$ such that $T_{\xi}r\in\mathcal{H}^\alpha([0,1]^d,\xi)$. Let $\mathcal{V}$ be the function class in \eqref{FNN_d_1} with the depth $\mathcal{D}=O(\log N_r)$, the weights $\mathcal{B}=O((N_r)^{d/(d+(4/\delta+2)\alpha)})$ and the size $\mathcal{S}=O((N_r)^{d/(d+(4/\delta+2)\alpha)}\log N_r)$. Setting $\xi\asymp(\kappa_{\delta}N_r/\mathcal{DS}\log \mathcal{S}\log N_r)^{1/(2+\delta)}$, the mean squared error of the estimator $\hat{r}_{\xi}(\cdot)$ in \eqref{hat_r_xi} satisfies
    \begin{center}
    $\mathbb{E}[\|\hat{r}_{\xi}-r\|_{L_2(P_{\boldsymbol{X}})}^2]\lesssim \kappa_{\delta}^{2/(2+\delta)}(N_r)^{-2\alpha/(d+(2+4/\delta)\alpha)}(\log N_r)^5$.
    \end{center}
\end{theorem}

Theorem \ref{Theorem r_xi_hat} provides a non-asymptotic error bound for the truncated density ratio estimator $\hat{r}_{\xi}(\cdot)$ in the unbounded scenario. As $\delta\to\infty$, the upper bound achieves the minimax optimal convergence rate $O((N_r)^{-2\alpha/(d+2\alpha)})$ up to logarithmic factors, recovering the minimax optimal rate derived in the uniformly bounded scenario. Based on the properties of $\hat{r}_{\xi}(\cdot)$, we next present the non-asymptotic error bound for $\hat{f}_{0,\hat{r}_{\xi}}(\cdot)$, which is obtained by substituting the truncated estimator $\hat{r}_{\xi}(\cdot)$ for $\hat{r}(\cdot)$ in \eqref{f_hat_r_hat}.

\begin{theorem}
    \label{Theorem Ef^_0_r^_zeta}
    Suppose that Assumptions \ref{assumption rsupp} and \ref{assumption rb0}-\ref{assumption fiepsilon} are satisfied and there exists some constant $\alpha>0$ such that $T_{\xi}r\in\mathcal{H}^\alpha([0,1]^d,\xi)$. Let $\mathcal{F}$ be the function class in \eqref{FNN_d_1} bounded by $B_1$ with the depth $\mathcal{D}=O(\log N_1^P)$, the weights $\mathcal{B}=O((N_1^P)^{{d}/(d+(2/(\delta+1)+2)\zeta)})$ and the size $\mathcal{S}=O((N_1^P)^{{d}/(d+(2/(\delta+1)+2)\zeta)}(\log N_1^P)^{-{5d}/(2\zeta+d}))$. Setting $\xi=O(({\mathcal{DS}\log\mathcal{S}(\log N_1^P)^{3}}/{N_1^P}+\\1/n_1^P)^{-{1}/(\delta+2)})$, then the mean squared error of the estimator $\hat{f}_{0,\hat{r}_{\xi}}(\cdot)$ obtained from \eqref{f_hat_r_hat} by replacing $\hat{r}(\cdot)$ with the truncated density ratio estimator $\hat{r}_{\xi}(\cdot)$ in \eqref{hat_r_xi} satisfies
\begin{align*}
    \mathbb{E}[\|\hat{f}_{0,\hat{r}_{\xi}}-f_0\|_{L_2(Q_{\boldsymbol{X}})}^2]& \lesssim (n_1^P)^{-(\delta+1)/(\delta+2)}+(N_1^P)^{-2\zeta/(d+(2/(\delta+1)+2)\zeta)}(\log N_1^P)^{{10\zeta}/(2\zeta+d)}\\
    & +\kappa_{\delta}^{2/(2+\delta)}(N_r)^{-{2\alpha}/(d+(2+4/\delta)\alpha)}(\log N_r)^5/M.
\end{align*}
Moreover, if $N_r\gtrsim M^{-\eta_{1}}{\kappa_{\delta}}^{\eta_2}(N_1^P)^{\eta_3}$ with $\eta_1=\{d\delta+2(\delta+2)\alpha\}/\{2\alpha\delta\}$, $\eta_2=\{\delta d+2(\delta+2)\alpha\}/\{\alpha\delta(\delta+2)\}$ and $\eta_3=\{2\zeta\alpha(\delta+2)(\delta+1)+d\zeta\delta(\delta+1)\}/\{2\zeta\alpha\delta(\delta+2)+d\alpha\delta(\delta+1)\}$, then we have
\begin{center} $\mathbb{E}[\|\hat{f}_{0,\hat{r}_{\xi}}-f_{0}\|_{L_{2}(Q_{X})}^{2}]\lesssim (n_1^P)^{-(\delta+1)/(\delta+2)}+(N_1^P)^{-2\zeta/(d+(2/(\delta+1)+2)\zeta)}(\log N_1^P)^{{10\zeta}/(2\zeta+d)}$.
\end{center}
\end{theorem}

Theorem \ref{Theorem Ef^_0_r^_zeta} provides a non-asymptotic errer upper bound for the estimator $\hat{f}_{0,\hat{r}_{\xi}}(\cdot)$ in the unbounded scenario. The upper bound consists of two terms, $(n_1^P)^{-(\delta+1)/(\delta+2)}+(N_1^P)^{-2\zeta/(d+(2/(\delta+1)+2)\zeta)}(\log N_1^P)^{{10\zeta}/(2\zeta+d)}$ captures the error due to estimating $f_0(\cdot)$, while $\kappa_{\delta}^{2/(2+\delta)}(N_r)^{-{2\alpha}/(d+(2+4/\delta)\alpha)}(\log N_r)^5/M$ accounts for the error resulted from estimating the density ratio $r(\cdot)$.
Moreover, Theorem \ref{Theorem Ef^_0_r^_zeta} implies that even when the density ratio is unbounded, if $N_r\gtrsim M^{-\eta_{1}}{\kappa_{\delta}}^{\eta_2}(N_1^P)^{\eta_3}$ and as $\delta\to\infty$, the truncated estimator $\hat{f}_{0,\hat{r}_{\xi}}(\cdot)$ attains the same minimax optimal rate $O((n_1^P)^{-1}+(N_1^P)^{-2\zeta/(2\zeta+d)})$ up to logarithmic factors as in the uniformly bounded case.

\subsection{When density ratio function $r(\cdot)$ is known}

We next investigate the non-asymptotic error bounds under the assumption that the density ratio $r(\cdot)$ is known. Recall that in the unknown density ratio scenario, we partition the source domain data into two disjoint subsets to construct $\hat{r}(\cdot)$ and to ensure the independence between the density ratio estimator and the estimator of $f_0(\cdot)$. In contrast, when $r(\cdot)$ is known, we no longer need to estimate $r(\cdot)$ so that we can utilize the entire source data to estimate $f_0(\cdot)$, thereby using more information and typically achieving tighter error bounds.
The following Theorems \ref{Theorem fhatr} and \ref{Theorem fhatr_xi} establish the results for the estimators $\hat{f}_{0,r}(\cdot)$ and $\hat{f}_{0,r_{\xi}}(\cdot)$ in the cases of uniformly bounded and unbounded density ratio, respectively. 

\begin{theorem}
    \label{Theorem fhatr}
    Suppose that Assumptions \ref{assumption rsupp}, \ref{assumption rb}, \ref{assumption f0B1} and \ref{assumption fiepsilon} are satisfied, and let $\mathcal{F}$ be the function class in \eqref{FNN_d_1} bounded by $B_1\geq1$ with the depth $\mathcal{D}=O(\log N^P)$, the weights $\mathcal{B}=O((N^P)^{d/(d+2\zeta)})$ and the size $\mathcal{S}=O((N^P)^{d/(2\zeta+d)}(\log N^P)^{-5d/(2\zeta+d)})$. Then, the mean squared error of the estimator $\hat{f}_{0,r}(\cdot)$ in \eqref{f_hat_r_hat1} satisfies
    \begin{center}
        $\mathbb{E}[\|\hat{f}_{0,r}-f_0\|_{L_2(Q_{\boldsymbol{X}})}^2] \lesssim \Gamma\big((n^P)^{-1}+(N^P)^{-2\zeta/(2\zeta+d)}(\log N^P)^{10\zeta/(2\zeta+d)}\big)$.
    \end{center}
\end{theorem}

Under the unbounded case, we introduce the truncated density ratio function below. For the same truncation level $\xi$ used in \eqref{hat_r_xi}, the truncated density ratio function $r_{\xi}(\cdot)$ is defined as
\begin{equation}
    r_{\xi}(\cdot) := \min\{r(\cdot), \xi\}
    \label{r_xi}
\end{equation}

\begin{theorem}
    \label{Theorem fhatr_xi}
    Suppose that Assumptions \ref{assumption rsupp}, \ref{assumption rub}, \ref{assumption f0B1} and \ref{assumption fiepsilon} are satisfied, and let $\mathcal{F}$ be the function class in \eqref{FNN_d_1}  bounded by $B_1\geq1$ with the depth $\mathcal{D}=O(\log N^P)$, the weights $\mathcal{B}=O((N^P)^{d/(d+(2/(\delta+1)+2)\zeta)})$ and the size $\mathcal{S}=O((N^P)^{d/(d+(2/(\delta+1)+2)\zeta)}(\log N^P)^{-5d/(2\zeta+d)})$. Let $\xi$ be chosen as $\xi=(1/(n^P\kappa_{\delta})+\mathcal{DS}\log\mathcal{S}(\log N^P)^{3}/(N^P\kappa_{\delta}))^{-1/(2+\delta)}$, then the mean squared error of the estimator $\hat{f}_{0,r_{\xi}}(\cdot)$ obtained from \eqref{f_hat_r_hat1} by replacing $r(\cdot)$ with the truncated density ratio $r_{\xi}(\cdot)$ in \eqref{r_xi} satisfies
    \begin{center}
    $\mathbb{E}\|\hat{f}_{0,r_{\xi}}-f_0\|_{L_2(Q_{\boldsymbol{X}})}^2\lesssim \kappa_{\delta}^{1/(2+\delta)}\big((n^P)^{-(\delta+1)/(\delta+2)}+(N^P)^{-2\zeta/(d+(2/(\delta+1)+2)\zeta)}(\log N^P)^{10\zeta/(2\zeta+d)}\big)$.
    \end{center}
\end{theorem}

Theorem \ref{Theorem fhatr} shows that under the uniformly bounded density ratio case, the estimator $\hat{f}_{0,r}(\cdot)$ also achieves the minimax optimal rate of $O((n^P)^{-1}+(N^P)^{-2\zeta/(2\zeta+d)})$ up to logarithmic factors.
Furthermore, Theorem \ref{Theorem fhatr_xi} extends this result to the unbounded case and establishes a non-asymptotic error bound for the truncated estimator $\hat{f}_{0,r_{\xi}}(\cdot)$ by selecting the truncation level as $\xi=(1/(n^P\kappa_{\delta})+\mathcal{DS}\log\mathcal{S}(\log N^P)^{3}/(N^P\kappa_{\delta}))^{-1/(2+\delta)}$ appropriately. Notably, as $\delta \to \infty$, the estimator $\hat{f}_{0,r_{\xi}}(\cdot)$ can also achieve the minimax optimal rate.

By comparing the results of Theorems \ref{Theorem Ef^_0_r^} and \ref{Theorem fhatr} as well as Theorems \ref{Theorem Ef^_0_r^_zeta} and \ref{Theorem fhatr_xi}, we observe that when the density ratio is unknown, the non-asymptotic error bound contains an additional term $O((N_r)^{-2\alpha/(d+2\alpha)}(\log N_r)^{5}/M)$ arising from estimating $r(\cdot)$. In addition, to ensure independence between the estimators of $r(\cdot)$ and $f_0(\cdot)$, we employ sample splitting, which reduces the effective sample size for estimating $f_0(\cdot)$ from $n^P$ to $n_1^P$. Consequently, in both the bounded and unbounded density ratio cases, the non-asymptotic error bounds under known $r(\cdot)$ are tighter than their corresponding bounds under unknown $r(\cdot)$.

Moreover, if each subject contributes only one observation, i.e., $m_i^P=m_i^Q=1$ so that $N^P=n^P$ and $N^Q=n^Q$, our dataset reduces to the independent observations setting. In this setting, our non-asymptotic error bounds are consistent with the minimax optimal rates for independent data up to logarithmic factors. To illustrate, we consider the case in Theorem \ref{Theorem fhatr} where the density ratio is known and uniformly bounded. According to Theorem \ref{Theorem fhatr}, the resulting error upper bound is of order $O((n^P)^{-2\zeta/(2\zeta+d)})$ up to logarithmic factors, which matches the existing minimax rate in the independent data setting \citep{10.1214/19-AOS1875, 10.1214/20-AOS2034}. Thus, the independent data setting is recovered as a special case of our analysis.

\section{Numerical Studies}
\label{sec: Numerical Studies}
\subsection{Simulations}

In this section, we evaluate the numerical performance of proposed estimators, including the estimators $\hat{f}_{0,\hat{r}}(\cdot)$ and $\hat{f}_{0,r}(\cdot)$ in \eqref{f_hat_r_hat} and \eqref{f_hat_r_hat1}, respectively. For comparison, we also consider a naive estimator $\hat{f}_0^{P}(\cdot)$, which is trained on all source domain data while ignoring the covariate shift between the source and target domains, defined as follows
\begin{equation} \label{hatf0P}
    \hat{f}_0^{P}(\cdot) \in \arg\min_{f \in \mathcal{F}} \frac{1}{N^P} \sum_{i=1}^{n^{P}} \sum_{j=1}^{m_i^P} (Y_{ij}^P- f(\boldsymbol{X}_{ij}^P))^2,
\end{equation}
where $N^P$ is the total number of observations in the source domain, and $\mathcal{F}$ is the function class in \eqref{FNN_d_1} bounded by $B_1$ with the depth $\mathcal{D}=O(\log N^P)$, the weights $\mathcal{B}=O((N^P)^{d/(d+2\zeta)})$ and the size $\mathcal{S}=O((N^P)^{d/(d+2\zeta)}(\log N^P)^{-5d/(d+2\zeta)})$. Under this choice, the estimator $\hat{f}_0^{P}(\cdot)$ for the target domain model achieves the minimax error rate when the density ratio is bounded, and achieves the sub-optimal rate when the density ratio is unbounded but has a finite second moment. The corresponding theoretical results in Theorem \ref{Theorem fPbounded} and \ref{Theorem fPunbounded} with proofs provided in \ref{f0P}. 

In the simulations, we denote the unknown ratio estimator $\hat{f}_{0,\hat{r}}(\cdot)$, the known ratio estimator $\hat{f}_{0,r}(\cdot)$, and the naive estimator $\hat{f}_0^{P}(\cdot)$ as URE, KRE, and NE, respectively. For each estimator, we consider both the bounded and unbounded density ratio settings. The implementation details for URE, KRE, and NE are provided below.
\begin{itemize}
    \item  
    \textit{URE for $\hat{f}_{0,\hat{r}}(\cdot)$:}
    For the estimation of $\hat{r}(\cdot)$ in \eqref{r_hat}, we specify $\mathcal{V}$ as the class of fully connected feedforward neural network with two hidden layers, using ReLU activations in the hidden layers and a Softplus activation in the output layer to remain nonnegativity of the estimated density ratio. The network is trained using Adam with an initial learning rate of $10^{-3}$, decayed by a factor of 0.5. 
    For $\hat{f}_{0,\hat{r}}(\cdot)$ in \eqref{f_hat_r_hat}, we specify $\mathcal{F}$ as the class of fully connected feedforward neural network with three hidden layers and ReLU activations. Training is carried out by stochastic gradient descent with nesterov momentum 0.9 and an initial learning rate of 0.01, which is decayed by a factor of 0.5.
    Training runs for at most 200 epochs with early stopping.
    
    \item \textit{KRE for $\hat{f}_{0,r}(\cdot)$:} 
    For the estimation of $\hat{f}_{0,r}(\cdot)$ in \eqref{f_hat_r_hat1}, we use the same network architecture, hyperparameter settings, and function class $\mathcal{F}$ as those used for $\hat{f}_{0,\hat{r}}(\cdot)$ in \eqref{f_hat_r_hat}.
    
    \item \textit{NE for $\hat{f}_0^{P}(\cdot)$:} 
    For the estimation of $\hat{f}_0^{P}(\cdot)$ in \eqref{hatf0P}, we use the same network architecture, hyperparameter settings, and function class $\mathcal{F}$ as those used in KRE.

\end{itemize}

For simplicity, we assume that the number of repeated measurements is the same for all subjects across all domains, that is, $m_i^P=m_i^Q=m$. 
A validation dataset $\mathcal{D}^{V}=\{(\boldsymbol{X}_{ij}^{V},Y_{ij}^{V}):1\leq i\leq n^{V}, 1\leq j\leq m_i^V\}$ for model selection is independently sampled from the source distribution, where $n^{V}=\lceil0.2n^{P} \rceil$ and $m_{i}^{V}=m$.
For the estimator URE, we set $n_{1}^{P}=n^{P}/2$, where the first $n_1^{P}$ subjects are used for target regression estimation and the other for density ratio estimation. In our study, we fix $n^Q=10000$ and report the mean squard error between $f_0(\cdot)$ and the estimators URE, KRE and NE together with their corresponding standard errors over 50 Monte Carlo replications under different scenarios, respectively. 

Recalling the model setting $Y_{i j}^Q=f_0(\boldsymbol{X}_{i j}^Q)+f_i(\boldsymbol{X}_{i j}^Q)+\varepsilon_{i j}$ in \eqref{model_Q}. We first consider the following low-dimensional case.

 \noindent \textbf{\textit {Case 1:}} We set the functions $f_0(\cdot)$ and $f_i(\cdot)$ as $f_0(\boldsymbol{X}_{ij}^{Q})=\sin(12\pi\sum_{l=1}^dlX_{ijl}/d(d+1))$ and $f_i(\boldsymbol{X}_{ij}^{Q})=\sum_{k=2}^\infty\sum_{l=1}^d\sqrt3\lambda_{ikl}(\sin(k\pi X_{ijl}^{Q})+\cos(k\pi X_{ijl}^{Q}))/\sqrt{d}k$, respectively.

In Case 1, 
$\boldsymbol{X}_{ij}^{Q}=(X_{ij1},\ldots,X_{ijd})^\top\in[0,1]^d$ with $d=3$.
We set $\boldsymbol{X}_{ij}^Q=\Phi(\boldsymbol{Z}_{ij}^Q)$, where $\boldsymbol{Z}_{ij}^Q\sim N(\mu_Q\mathbf{1}_d,\sigma_{Q}^{2}\mathbf{I}_d)$, and $\Phi(\boldsymbol{Z}_{ij}^Q)$ denotes applying the cumulative distribution function of the standard normal distribution componentwise to $\boldsymbol{Z}_{ij}^Q$.
Similarly, we let $\boldsymbol{X}_{ij}^P=\Phi(\boldsymbol{Z}_{ij}^P)$ and $\boldsymbol{Z}_{ij}^P\sim N(\mu_{P}\mathbf{1}_d,\sigma_P^{2}\mathbf{I}_d)$.
It is easy to show that the density ratio $r(\cdot)$ is uniformly bounded when $\sigma_P^{2}\geq\sigma_Q^{2}$, and unbounded but has a finite second moment when $\sigma_Q^{2}\geq\sigma_P^{2}\geq\sigma_Q^{2}/2$. In particular, we use $\mu_{P}=0,\sigma_P^{2}=0.4,\mu_Q=0.5,\sigma_Q^{2}=0.3$ for the bounded density ratio setting and $\mu_{P}=0,\sigma_P^{2}=0.3,\mu_{Q}=1,\sigma_Q^{2}=0.5$ for the setting in which the density ratio is unbounded but has a finite second moment. The noise terms are independently generated from $\varepsilon_{ij}\sim N(0,0.01^2)$ and the coefficients $\lambda_{ikl}$ are independently drawn from $N(0,0.1^2)$.

\begin{table}[!htbp]
\centering
\caption{Prediction mean squared errors with standard errors for Case 1. Here all entries are scaled by $10^{-1}$.}
\label{table:d=3}
\begin{tabular}{c c c c c c}
\toprule
& & \multicolumn{2}{c}{bounded $r(\cdot)$} & \multicolumn{2}{c}{unbounded $r(\cdot)$} \\
\cmidrule(lr){3-4}\cmidrule(lr){5-6}
$m$ & Method & $n^P=500$ & $n^P=1000$
            & $n^P=500$ & $n^P=1000$ \\
\midrule
\multirow{3}{*}{25}
& NE   & 0.041 (0.028) & 0.030 (0.032) & 0.280 (0.251) & 0.108 (0.080) \\
& KRE  & 0.024 (0.009) & 0.017 (0.008) & 0.065 (0.023) & 0.064 (0.022)\\
& URE  & 0.383 (0.294) & 0.087 (0.028) & 0.076 (0.033) & 0.066 (0.030)\\
\midrule
\multirow{3}{*}{50}
& NE   & 0.031 (0.020) & 0.021 (0.014) & 0.117 (0.091) & 0.060 (0.026)\\
& KRE  & 0.020 (0.011) & 0.012 (0.005) & 0.063 (0.021) & 0.056 (0.014)\\
& URE  & 0.030 (0.013) & 0.019 (0.008) & 0.065 (0.026) & 0.058 (0.024)\\
\midrule
\multirow{3}{*}{75}
& NE   & 0.022 (0.013) & 0.016 (0.006) & 0.071 (0.038) & 0.051 (0.022)\\
& KRE  & 0.016 (0.007) & 0.011 (0.005) & 0.058 (0.015) & 0.050 (0.012)\\
& URE  & 0.022 (0.011) & 0.015 (0.005) & 0.059 (0.031) & 0.050 (0.015)\\
\bottomrule
\end{tabular}
\end{table}

Table \ref{table:d=3} presents the prediction mean squared error (MSE) and corresponding standard errors, scaled by $10^{-1}$, for the three estimators NE, KRE and URE under Case 1 across varying $n^P$ and $m$. For a generic estimator $\hat{f}(\cdot)$, the prediction MSE is defined as $\text{MSE}(\hat{f}) = \sum_{i=1}^{n^Q}\sum_{j=1}^m (\hat{f}(\boldsymbol{X}^Q_{ij}) - f_0(\boldsymbol{X}^Q_{ij}))^2/(n^Qm)$ . Overall, the MSEs of all three estimators decrease as both $n^P$ and $m$ increase, aligning with theoretical expectations of consistency. Under the bounded density ratio setting, when $m=25$, the KRE estimator demonstrates the best performance, whereas URE yields the largest MSE, even underperforming the NE estimator. 
This relative inefficiency of URE on the one hand may be due to the inherent sample splitting mechanism, leading to a noticeable loss of estimation accuracy when $m$ is relatively small. On the other hand, the substantial overlap in covariate support between the source and target domains significantly mitigates the impact of covariate shift, allowing the NE to maintain predictive accuracy even without accounting for the shift.
As $m$ increases to 50 and 75, KRE consistently maintains the lowest MSE, while the MSEs of NE and URE converge to comparable levels. However, URE achieves strictly smaller standard errors compared to NE under these larger $m$ settings, indicating that it provides more stable estimates while maintaining comparable prediction accuracy.

The results under the unbounded density ratio setting further highlight the advantages of the proposed methods. As expected, the KRE, which utilizes the true known density ratio, consistently yields the minimum MSE across all combinations of $n^P$ and $m$. 
Furthermore, in contrast to the bounded setting, URE uniformly outperforms the NE estimator in terms of MSE. 
By effectively correcting for the covariate shift, both the proposed KRE and URE methods provide substantially more accurate and robust predictions than the NE method.

Next to further evaluate the predictive performances of the estimators, Case 2 considers an underlying low-dimensional manifold structure with $d=10$. The covariates $\boldsymbol{X}_{ij}^Q$ and $\boldsymbol{X}_{ij}^P$ are generated following the identical procedure described in Case 1.
Under this manifold structure, the $f_0(\cdot)$ and $f_i(\cdot)$ are specified as follows.

\noindent
\textbf{\textit {Case 2:}} 
We set
$f_0(\boldsymbol{X}^{Q}_{ij})=\sin(2\pi \sum_{l=1}^{5}X^{Q}_{ijl}/5)\cos(2\pi \sum_{l=6}^{10}X^{Q}_{ijl}/5)$ and
\vspace{-0.5cm}
\begin{align*}
f_i(\boldsymbol{X}_{ij}^{Q})=&\sum_{k=2}^\infty\sqrt{3}\{\lambda_{ik1}[ \sin(k\pi (\sum_{l=1}^5 X^Q_{ijl}/5))+\cos(k\pi (\sum_{l=1}^5X^Q_{ijl}/5))]\\
&+\lambda_{ik2}[\sin(k\pi (\sum_{l=6}^{10}X^Q_{ijl}/5))+\cos(k\pi (\sum_{l=6}^{10}X^Q_{ijl}/5))]\}/\sqrt{2}k.
\end{align*}

Table 2 summarizes the simulation results for Case 2.
Similar to the results in Case 1, the prediction MSEs of all estimators decrease as $n^P$ and $m$ increase, confirming the numerical consistency of the methods. 
Notably, the absolute values of the MSEs are larger than most of those in Case 1, especially under the unbounded density ratio setting, due to the increased complexity and dimensionality of the model.
\begin{table}[!htbp]
\centering
\caption{Prediction mean squared errors with their standard errors for Case 2.}
\label{table:d=10}
\begin{tabular}{c c c c c c}
\toprule
& & \multicolumn{2}{c}{bounded $r(\cdot)$} & \multicolumn{2}{c}{unbounded $r(\cdot)$} \\
\cmidrule(lr){3-4}\cmidrule(lr){5-6}
$m$ & Method & $n^P=500$ & $n^P=1000$
            & $n^P=500$ & $n^P=1000$ \\
\midrule
\multirow{3}{*}{25}
& NE   & 0.078 (0.018) & 0.046 (0.016) & 1.194 (0.160) & 0.929 (0.164) \\
& KRE  & 0.032 (0.010) & 0.026 (0.011) & 0.595 (0.124) & 0.416 (0.120) \\
& URE  & 0.048 (0.014) & 0.043 (0.015) & 0.625 (0.122) & 0.548 (0.087) \\
\midrule
\multirow{3}{*}{50}
& NE   & 0.045 (0.015) & 0.028 (0.008) & 0.900 (0.134) & 0.712 (0.099) \\
& KRE  & 0.025 (0.011) & 0.015 (0.011) & 0.455 (0.154) & 0.310 (0.091) \\
& URE  & 0.041 (0.013) & 0.026 (0.005) & 0.468 (0.104) & 0.431 (0.058) \\
\midrule
\multirow{3}{*}{75}
& NE   & 0.037 (0.013) & 0.023 (0.007) & 0.789 (0.109) & 0.645 (0.117) \\
& KRE  & 0.018 (0.010) & 0.007 (0.008) & 0.351 (0.105) & 0.231 (0.062) \\
& URE  & 0.027 (0.006) & 0.023 (0.004) & 0.448 (0.083) & 0.401 (0.053) \\
\bottomrule
\end{tabular}
\end{table}

Under the bounded setting, the URE estimator demonstrates a more apparent advantage over the NE estimator compared to the low-dimensional setting. Unlike Case 1, where NE was competitive at $m=25$, URE now uniformly outperforms NE even with a small $m$. This suggests that as the dimensionality increases, the advantage of correcting for covariate shift via URE outweighs the efficiency loss from sample splitting more rapidly. As $m$ increases to 50 and 75, KRE continues to attain the lowest MSE, while URE maintains its lead over NE with consistently smaller standard errors.
Under the unbounded $r(\cdot)$ setting, KRE uniformly attains the lowest MSE and URE consistently outperforms NE. 
Although the MSE of all estimators decreases monotonically as $n^P$ and $m$ increase, the relative performance disparity between NE and the other two estimators persists and is significantly larger than that in Case 1. This highlights the necessity of correcting for the covariate shift.

\subsection{Real Data Analysis}
We evaluate the proposed methods using an air quality dataset obtained from \citet{zheng2024dynamic}. The dataset consists of observations from 94 monitoring stations located in Beijing, Shanghai, Shandong, Hebei, and Tianjin. We treat the 20 stations in Beijing as the target domain and the remaining stations in the other regions as the source domain. Each station is viewed as a subject, with hourly measurements of air pollutant concentrations and meteorological variables collected over time. We take the logarithm of the PM2.5 concentration as the response variable. The eight covariates consist of the concentrations of SO$_2$, NO$_2$, O$_3$ and CO together with the wind speed, humidity, dew point temperature, and air pressure.

For each station, measurements are continuously recorded over 72 hours, and the final 24 hours fall within the period of the Beijing orange air pollution alert, which affects all stations in Beijing. To mitigate the impact of the alert and make the assumption ${P}_{Y|\boldsymbol{X}}={Q}_{Y|\boldsymbol{X}}$ more plausible, we retain only the first 48 hours of observations for each station. To handle missing data, we delete hourly observations with missing values. In the target domain, we further remove the entire station if any variable is missing. After preprocessing, 11 stations remain in the target domain and 64 in the source domain. 

Since our methods are designed to address covariate shift, we first examine whether the source and target domains differ in their covariate distributions. Figure \ref{fig:covariate_shift_hist} shows the normalized histograms and smoothed empirical density curves of each covariate across the two domains. 
Because the distributions of SO$_2$, O$_3$, wind speed and CO have substantial right skewness, we apply the $h(x)=\log(1+x)$ transformation to these four covariates in order to facilitate a clearer comparison of their distributions across the two domains. 
As shown in Figure \ref{fig:covariate_shift_hist}, across the source and target domains, covariates SO$_2$, air pressure, dew point temperature, O$_3$, NO$_2$, and wind speed exhibit distinct distributional differences, while the distributions of humidit and CO are similar, suggesting that the two domains do not arise from the same covariate distribution.
As the density ratio is unknown in practice, we next evaluate the predictive performance of the NE and URE estimators on the target domain.
 \begin{figure}[htbp]
    \centering
    \includegraphics[width=\textwidth]{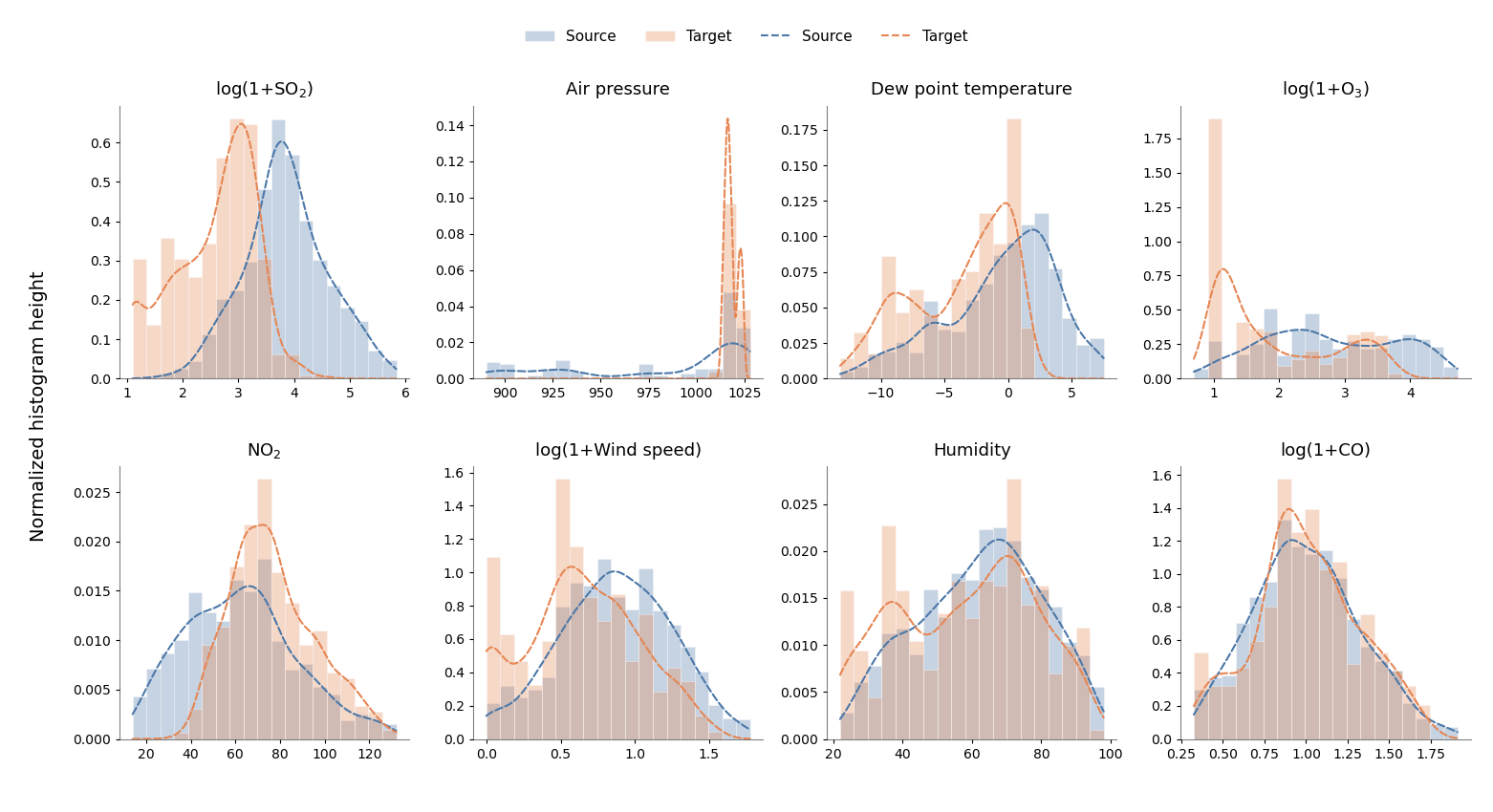}
    \caption{Normalized histograms and smoothed empirical density curves of the eight covariates in the source and target domains. 
    }
    \label{fig:covariate_shift_hist}
\end{figure}
\begin{figure}[htbp]
    \centering
    \begin{subfigure}[t]{0.37\textwidth}
        \centering
        \includegraphics[width=\textwidth]{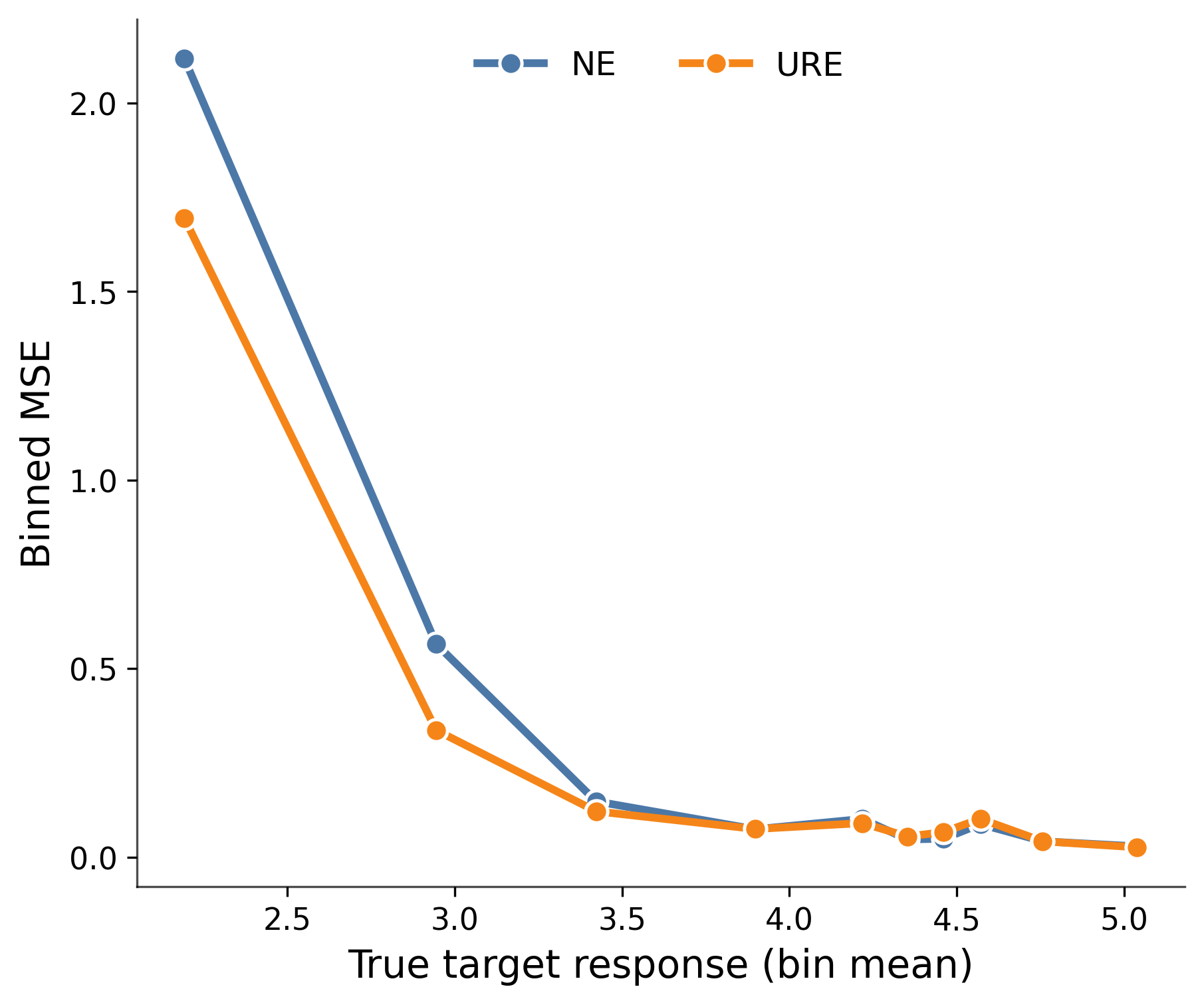}
        \caption{}
        \label{fig:binned_mse}
    \end{subfigure}
    \hspace{0.02\textwidth}
    \begin{subfigure}[t]{0.37\textwidth}
        \centering
        \includegraphics[width=\textwidth]{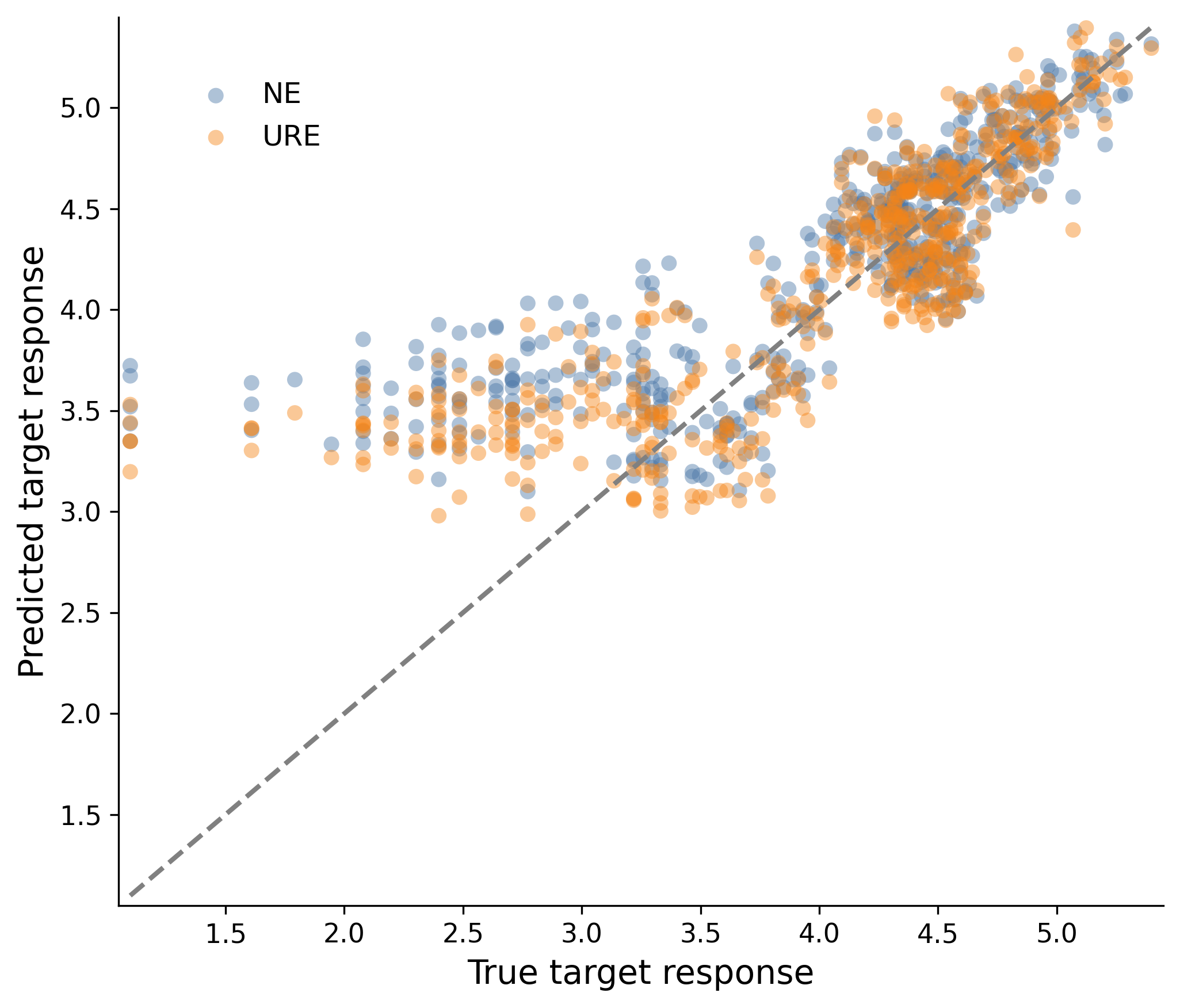}
        \caption{}
        \label{fig:true_pred_overlay}
    \end{subfigure}
    \caption{Predictive performance of the estimators NE and URE on the target domain. (a) Binned MSE with respect to the mean true target response in each bin. (b) True versus predicted responses for the two estimators.} 
    \label{fig:NE_URE_compare}
\end{figure}

Figure \ref{fig:NE_URE_compare} summarizes the predictive performance of the estimators NE and URE on the target domain. 
In Figure \ref{fig:NE_URE_compare}(\subref{fig:binned_mse}), we report the binned mean squared error (binned MSE) with respect to the true target response within each bin. Specifically, the samples are first partitioned into ten bins according to the quantiles of the true responses. 
The x-axis represents the average true response within each bin. For a given bin, the corresponding binned MSE is computed as the mean squared error between the predicted and true responses over all samples assigned to that bin. As shown in Figure \ref{fig:NE_URE_compare}(\subref{fig:binned_mse}), the estimator URE achieves a smaller binned MSE than NE over most of the target response range, with the improvement being particularly significant at lower target responses. 
Figure \ref{fig:NE_URE_compare}(\subref{fig:true_pred_overlay}) further illustrates predictive performance through a scatter plot of the true responses versus the predicted responses for both estimators. The dashed line means the predicted response is equal to the true response so that points lying closer to this line indicate more accurate predictions. As shown in Figure \ref{fig:NE_URE_compare}(\subref{fig:true_pred_overlay}), points predicted by estimator URE are generally more concentrated around the dashed line than those of NE, especially in the low target response region. 
Overall, these findings are consistent with the simulation results and show the advantange of correcting for the covariate shift in prediction of target domain.

\section{Conclusion and Discussion}
We have developed a deep transfer learning framework for repeated measurements regression under covariate shift when target domain responses are unavailable or expensive to collect. The proposed approach used density ratio weighting to transfer information from a labeled source domain to an unlabeled target domain, with both the density ratio and the target regression function estimated by ReLU FNNs. For both unknown and known density
ratio, and for both uniformly bounded and unbounded density ratio settings, we
established non-asymptotic error bounds and showed that the proposed estimators could achieve
the minimax optimal convergence rate up to logarithmic factors under suitable conditions.
We also recovered the independent data problem as a special case of the repeated measurements framework. 
There remain possible promising directions for future research, including the incorporation of multiple source datasets, the use of partially observed responses in the target domain and the development of more adaptive approaches for handling unbounded density ratios.

\clearpage

\bibliographystyle{apalike_copy}
\bibliography{./Deep_Regression_for_Repeated_Measurements_under_Covariate_Shift}
\newpage
\appendix

\label{Appendix}
\renewcommand{\appendixname}{Appendix}
\renewcommand{\thesection}{\appendixname~\Alph{section}}
\renewcommand{\thesubsection}{\Alph{section}.\arabic{subsection}}
\numberwithin{equation}{section}
\renewcommand{\theequation}{\Alph{section}.\arabic{equation}}
\label{Appendix} 
\section{Overview}

This appendix provides detailed proofs
of the main theoretical results and additional technical results.
The remaining material is organized as follows. 

\begin{itemize}

    \item \ref{proofs} includes the proofs for the density ratio estimators and the target regression estimators based on density ratio in Section \ref{sec: Theoretical Analysis}. 
    Section \ref{Proofr} studies the unknown density ratio estimation and proves Theorems \ref{Theorem rstochastic}, \ref{Theorem approx}, \ref{Theorem 3rb} and \ref{Theorem r_xi_hat}. 
    Sections \ref{ProofURE} and \ref{ProofKRE} contain the proofs of Theorems \ref{Theorem Ef^_0_r^} and \ref{Theorem Ef^_0_r^_zeta}, and Theorems \ref{Theorem fhatr} and \ref{Theorem fhatr_xi}, respectively, which establish the non-asymptotic error bounds for the target regression estimators when the density ratio is unknown and known. 
    
    \item \ref{f0P} develops a theoretical analysis of the naive estimator $\hat{f}_{0}^{P}(\cdot)$ in \eqref{hatf0P}. 
    It provides non-asymptotic error bounds under both bounded density ratio and finite second moment assumptions, together with the corresponding proofs.
\end{itemize}

\section{Proof of the Main Results in Section \ref{sec: Theoretical Analysis}}\label{proofs}

We first introduce the definition of the covering number, which quantifies the complexity of the candidate function class and will play a key role in the subsequent proofs.

\begin{definition}[Covering Number] For any $\epsilon >0$, the $\epsilon$-covering number of a function class $\mathcal{G}$ with respect to the metric $\|\cdot\|_{L_{2}(P_{\boldsymbol{X}})}$ is defined as
\begin{equation*}
    \mathcal{N}{ (\epsilon,\mathcal{G},\|\cdot\|_{L_{2}(P_{\boldsymbol{X}})})}:=\min\big\{m\in\mathbb{N}:\exists g_1,\ldots,g_m\in\mathcal{G}\mathrm{~s.t.~}\forall g\in\mathcal{G},\min_{1\leq j\leq m}\|g-g_j\|_{L_2(P_{\boldsymbol{X}})}\leq\epsilon\big\}.
\end{equation*}
\end{definition}\label{Covering number}
\subsection{Proofs for Unknown Density Ratio Estimation}\label{Proofr}

In this section, we present the proofs of Theorems \ref{Theorem rstochastic}, \ref{Theorem approx}, \ref{Theorem 3rb} and \ref{Theorem r_xi_hat}, which are used
to establish the theoretical properties of density ratio estimation. Theorem \ref{Theorem rstochastic} provides a stochastic error bound for the density ratio estimator when density ratio is uniformly bounded, while Theorem \ref{Theorem approx} establishes the approximation error bound for ReLU FNNs over H\"older function classes. Based on these two results, Theorem \ref{Theorem 3rb} derives a non-asymptotic upper bound for the density ratio estimator under a uniformly bounded density ratio setting. Theorem \ref{Theorem r_xi_hat} further extends the analysis to the unbounded density ratio setting.
\subsubsection{Proof of Theorem \ref{Theorem rstochastic}}
\begin{proof}
 Define $\tilde{\ell}_{P}(v;\boldsymbol{X}):=0.5v^2(\boldsymbol{X})-0.5r^2(\boldsymbol{X})$ and $\tilde{\ell}_{Q}(v;\boldsymbol{X}):=v(\boldsymbol{X})-r(\boldsymbol{X})$, then the stochastic error term $\mathbb{E}[L(r)-2\widehat{L}(\hat{r})+L(\hat{r})]$ can be written as
\begin{align*}
    &\mathbb{E}[L(r)-2\widehat{L}(\hat{r})+L(\hat{r})]=\mathbb{E}[L(\hat{r})-L(r)-2(\widehat{L}(\hat{r})-\widehat{L}(r))]\\
    =&\mathbb{E}[\tilde{\ell}_{P}(\hat{r};\boldsymbol{X}_{ij}^{P})-\tilde{\ell}_{Q}(\hat{r};\boldsymbol{X}_{ij}^{Q})]-2\mathbb{E}[\sum_{i=n_1^{P}+1}^{n^P}\sum_{j=1}^{m_{i}^{P}}\tilde{\ell}_{P}(\hat{r};\boldsymbol{X}_{ij}^{P})/N_{2}^{P}-\sum_{i=1}^{n^Q}\sum_{j=1}^{m_i^Q}\tilde{\ell}_{Q}(\hat{r};\boldsymbol{X}_{ij}^{Q})/N^{Q}].
\end{align*}
Let $\mathcal{D}^{P^{\prime}}_{2}:=\{(\boldsymbol{X}_{ij}^{P^{\prime}},Y_{ij}^{P^{\prime}}):n_1^P+1\leq i\leq n^{P}, 1\leq j\leq m_i^P\}$ and $\mathcal{D}^{Q^{\prime}}:=\{\boldsymbol{X}_{ij}^{Q^{\prime}}: 1\leq i\leq n^{Q}, 1\leq j\leq m_i^{Q}\}$ be independent copies of $\mathcal{D}^{P}_{2}$ and $\mathcal{D}^{Q}$, respectively. Let $G_{P}(v;\boldsymbol{X}_{ij}^{P}):=\mathbb{E}_{\mathcal{D}_{2}^{P^{\prime}}}[\tilde{\ell}_{P}(v;\boldsymbol{X}_{ij}^{P^{\prime}})-2\tilde{\ell}_{P}(v;\boldsymbol{X}_{ij}^{P})]$ and $G_{Q}(v;\boldsymbol{X}_{ij}^{Q}):=\mathbb{E}_{\mathcal{D}^{Q^{\prime}}}[\tilde{\ell}_{Q}(v;\boldsymbol{X}_{ij}^{Q^{\prime}})-2\tilde{\ell}_{Q}(v;\boldsymbol{X}_{ij}^{Q})]$. Then we have
\begin{equation}
    \mathbb{E}[L(r)-2\widehat{L}(\hat{r})+L(\hat{r})]=\mathbb{E}[\sum_{i=n_1^{P}+1}^{n^P}\sum_{j=1}^{m_{i}^{P}}G_{P}(\hat{r};\boldsymbol{X}_{ij}^{P})/N_{2}^{P}-\sum_{i=1}^{n^Q}\sum_{j=1}^{m_i^Q}G_{Q}(\hat{r};\boldsymbol{X}_{ij}^{Q})/N^{Q}].
    \label{T1Gresult}
\end{equation}
Let $\mathcal{N}_{\mathcal{V}}:=\mathcal{N}(\epsilon,\mathcal{V},\|\cdot\|_{L_{2}(P_{\boldsymbol{X}})})$ be the $\epsilon$-covering number of $\mathcal{V}$ with respect to $\|\cdot\|_{L_2(P_X)}$. Accordingly, there exists an $\epsilon$-cover $\mathcal{C}_{\mathcal{V}}=\{v_1,\ldots,v_{\mathcal{N}_{\mathcal{V}}}\}\subset\mathcal{V}$. Then it follows that for any $v\in\mathcal{V}$, there exists a $\tilde{v}\in\mathcal{C}_{\mathcal{V}}$ such that 
\begin{center}
    $\mathbb{E}[\tilde{\ell}_{P}(v;\boldsymbol{X}_{i j}^{P})-\tilde{\ell}_{P}(\tilde{v};\boldsymbol{X}_{i j}^{P})]\leq \Gamma\epsilon,~\mathbb{E}[\tilde{\ell}_{Q}(v;\boldsymbol{X}_{i j}^{Q})-\tilde{\ell}_{Q}(\tilde{v};\boldsymbol{X}_{i j}^{Q})]\leq \Gamma\epsilon,$
\end{center}
\begin{center}
    $\mathbb{E}[G_{P}(v;\boldsymbol{X}_{ij}^{P})-G_{P}(\tilde{v};\boldsymbol{X}_{ij}^{P})]\leq3\Gamma\epsilon,~ \mathbb{E}[G_{Q}(v;\boldsymbol{X}_{ij}^{Q})-G_{Q}(\tilde{v};\boldsymbol{X}_{ij}^{Q})]\leq3\Gamma\epsilon.$
\end{center}

It is easy to show that $|\tilde{\ell}_{P}(\tilde{v};\boldsymbol{X}_{ij}^{P})-\mathbb{E}[\tilde{\ell}_{P}(\tilde{v};\boldsymbol{X}_{ij}^{P})]|\leq4\Gamma^2$, and $|\tilde{\ell}_{Q}(\tilde{v};\boldsymbol{X}_{ij}^{Q})-\mathbb{E}[\tilde{\ell}_{Q}(\tilde{v};\boldsymbol{X}_{ij}^{Q})]|\\\leq4\Gamma$. Denote by $b:=4\Gamma^2/N_2^P+4\Gamma/N^Q$, $A_1:=\sum_{i=n_1^P+1}^{n^P}\sum_{j=1}^{m_i^P}\tilde{\ell}_{P}(\tilde{v};\boldsymbol{X}_{ij}^{P})/N_2^{P}$, $A_2:=-\sum_{i=1}^{n^Q}\sum_{j=1}^{m_i^Q}\tilde{\ell}_{Q}(\tilde{v};\boldsymbol{X}_{ij}^{Q})/N^{Q}$, $A=A_1+A_2$ and $\sigma^2:=Var(A)$, then we have
\begin{align*}
    \sigma^2&=Var(\tilde{\ell}_{P}(\tilde{v};\boldsymbol{X}_{ij}^{P}))/N_2^{P}+Var(\tilde{\ell}_{Q}(\tilde{v};\boldsymbol{X}_{ij}^{Q}))/N^{Q}\\
    &\leq\mathbb{E}[(\tilde{\ell}_{P}(\tilde{v};\boldsymbol{X}_{ij}^{P}))^{2}]/N_2^{P}+\mathbb{E}[(\tilde{\ell}_{Q}(\tilde{v};\boldsymbol{X}_{ij}^{Q}))^{2}]/N^{Q}\\
    &=\mathbb{E}[(\tilde{v}(\boldsymbol{X}_{ij}^P)-r(\boldsymbol{X}_{ij}^P))^{2}(\tilde{v}(\boldsymbol{X}_{ij}^P)+r(\boldsymbol{X}_{ij}^P))^{2}]/4N_2^{P}+\mathbb{E}[(\tilde{v}(\boldsymbol{X}_{ij}^Q)-r(\boldsymbol{X}_{ij}^Q))^2]/N^{Q}\\
    & \leq(\Gamma^2/N_2^P+\Gamma/N^Q)\mathbb{E}[(\tilde{v}(\boldsymbol{X}_{ij}^P)-r(\boldsymbol{X}_{ij}^P))^2]\\
    & =(2\Gamma^2/N_2^P+2\Gamma/N^Q)\mathbb{E}[\tilde{\ell}_{P}(\tilde{v};\boldsymbol{X}_{ij}^P)-\tilde{\ell}_{Q}(\tilde{v};\boldsymbol{X}_{ij}^Q)]=(2\Gamma^2/N_2^P+2\Gamma/N^Q)\mathbb{E}[A].
\end{align*}
Thus, we have $\mathbb{E}[A]\geq \sigma^2N_{2}^{P}N^{Q}/(2\Gamma(N_2^P+\Gamma N^Q))$. 
Let $u:=t/2+\sigma^2N_{2}^{P}N^{Q}/(4\Gamma(N_2^P+\Gamma N^Q))$ with $t>0$, it is easy to show that $\sigma^2/u\leq4\Gamma^2/N_2^{P}+4\Gamma/N^{Q}$ and  $u\geq t/2$. By Bernstein concentration inequality, we have
\begin{align*}
 &\mathbb{P}(\sum_{i=n_1^P+1}^{n^P}\sum_{j=1}^{m_i^P}G_{P}(\tilde{v};\boldsymbol{X}_{ij}^{P})/N_{2}^{P}-\sum_{i=1}^{n^Q}\sum_{j=1}^{m_i^Q}G_{Q}(\tilde{v};\boldsymbol{X}_{ij}^{Q})/N^{Q}>t)\\
 &~~~\leq\mathbb{P}(\mathbb{E}[A]-A>t/2+\sigma^2N_{2}^{P}N^{Q}/(2\Gamma(N_2^P+\Gamma N^Q)))=\mathbb{P}\left(\mathbb{E}[A]-A>u\right) \\
  &~~~\leq\exp(-u^2/2(\sigma^2+bu))\leq\exp\left(-N_ru/(16\Gamma(\Gamma+1))\right) \\
  &~~~\leq\exp(-N_rt/(32\Gamma(\Gamma+1))),
\end{align*}
where $1/N_r=1/N_2^P+1/N^Q$. Setting $a=6\Gamma\epsilon+32\Gamma(\Gamma+1)\log\mathcal{N}_{\mathcal{V}}/N_r$ with $\epsilon=(\Gamma+1)/N_r$, we have
    \begin{align}
    &\mathbb{E}[L(r)-2\widehat{L}(\hat{r})+L(\hat{r})] =\mathbb{E}[\sum_{i=n_1^{P}+1}^{n^P}\sum_{j=1}^{m_{i}^{P}}G_{P}(\hat{r};\boldsymbol{X}_{ij}^{P})/N_{2}^{P}-\sum_{i=1}^{n^Q}\sum_{j=1}^{m_{i}^{Q}}G_{Q}(\hat{r};\boldsymbol{X}_{ij}^{Q})/N^{Q}]\nonumber\\
     \leq&\mathbb{E}[\max_{v\in\mathcal{V}}\{\sum_{i=n_1^{P}+1}^{n^P}\sum_{j=1}^{m_{i}^{P}}G_{P}(v;\boldsymbol{X}_{ij}^{P})/N_{2}^{P}-\sum_{i=1}^{n^Q}\sum_{j=1}^{m_{i}^{Q}}G_{Q}(v;\boldsymbol{X}_{ij}^{Q})/N^{Q}\}]\nonumber\\
     \leq&\mathbb{E}[\max_{\tilde{v}\in\mathcal{C}_{\mathcal{V}}}\{\sum_{i=n_1^{P}+1}^{n^P}\sum_{j=1}^{m_{i}^{P}}G_{P}(\tilde{v};\boldsymbol{X}_{ij}^{P})/N_{2}^{P}-\sum_{i=1}^{n^Q}\sum_{j=1}^{m_{i}^{Q}}G_{Q}(\tilde{v};\boldsymbol{X}_{ij}^{Q})/N^{Q}\}+6\Gamma\epsilon]\nonumber\\
     \leq& a+\int_{a}^{\infty}\mathcal{N}_{\mathcal{V}}\exp(-N_r(t-6\Gamma\epsilon)/(32\Gamma(\Gamma+1)))\mathrm{d}t\nonumber \\
     \leq & a+32\mathcal{N}_{\mathcal{V}}\Gamma(\Gamma+1)\exp(-N_r(a-6\Gamma\epsilon)/(32\Gamma(\Gamma+1)))/N_r\nonumber\\
     =&6\Gamma\epsilon+32\Gamma(\Gamma+1)(\log\mathcal{N}_{\mathcal{V}}+1)/N_r \leq 32\Gamma(\Gamma+1)(\log\mathcal{N}_{\mathcal{V}}+2)/N_r.
    \label{T1IlogN}
\end{align}
Applying Theorem 9.4 in \citet{gyorfi2002distribution}, we have
\begin{equation}
    \log\mathcal{N}((\Gamma+1)/N_r,\mathcal{V},\|\cdot\|_{L_{2}(P_{\boldsymbol{X}})})\lesssim\mathrm{VCdim}(\mathcal{V})\left(\log 2eN_{r}+\log\log 3eN_{r}\right)+\log3.
    \label{logN_VC}
\end{equation}
By Theorem 6 of \citet{bartlett2019nearly}, the VC dimension of a ReLU network with  depth $\mathcal{D}$ and  size $\mathcal{S}$ satisfies
\begin{equation}
    \mathrm{VCdim}\lesssim \mathcal{D}\mathcal{S}\log(\mathcal{S}).
    \label{VC_LW}
\end{equation}
Hence, combining \eqref{T1IlogN}, \eqref{logN_VC} and \eqref{VC_LW}, we obtain
\begin{equation}
    \mathbb{E}[L(r)-2\widehat{L}(\hat{r})+L(\hat{r})]\lesssim\Gamma(\Gamma+1)\mathcal{DS}\log\mathcal{S}\log N_r/N_r. 
    \label{T1(I)result}
\end{equation}
Thus we complete the proof.
\end{proof}

\subsubsection{Proof of Theorem \ref{Theorem approx}}\label{AppendixA.1.2}
Before giving the proof of Theorem \ref{Theorem approx}, we first introduce Lemma \ref{Lemma 3}, whose proof follows arguments similar to Lemma A.8 of \citet{PETERSEN2018296}.
It primarily illustrates the relationship between the remainder term of the Taylor expansion of $f$ at $\boldsymbol{x_0}$ and the Hölder smoothness with respect to the infinity norm $\|\boldsymbol{x} - \boldsymbol{x_0}\|_\infty$. 
\begin{lemma} 
\label{Lemma 3} 
For any $\boldsymbol{x},\boldsymbol{x_0} \in [0, 1]^d$, $f \in \mathcal{H}^\zeta([0,1]^d, B)$ and $\zeta$ serves as the Hölder smoothness index, the following inequality holds:
\begin{equation*}
    | f(\boldsymbol{x}) - \sum_{\|\boldsymbol{\alpha}\|_1 \leq t} \partial^{\boldsymbol{\alpha}} f(\boldsymbol{x_0}) (\boldsymbol{x} - \boldsymbol{x_0})^{\boldsymbol{\alpha}}/\boldsymbol{\alpha}! | \leq Bd^t \|\boldsymbol{x} - \boldsymbol{x_0}\|_{\infty}^{\zeta}.
\end{equation*}
\end{lemma}
\begin{proof}
We denote $g(u) = f(\boldsymbol{x}_0+u\boldsymbol{h})$, where $\boldsymbol{h} = \boldsymbol{x}-\boldsymbol{x}_0$ and $u \in [0,1]$. By repeated applications of the chain rule, we have $g^{(k)}(u) = \sum_{\|\boldsymbol{\alpha}\|_1 = k} k!h^{\boldsymbol{\alpha}}\partial^{\boldsymbol{\alpha}} f(\boldsymbol{x}_0+u\boldsymbol{h})/\boldsymbol{\alpha}!$, therefore
\begin{equation}
    | f(\boldsymbol{x}) - \sum_{\|\boldsymbol{\alpha}\|_1 \leq t} \partial^{\boldsymbol{\alpha}} f(\boldsymbol{x_0}) (\boldsymbol{x} - \boldsymbol{x_0})^{\boldsymbol{\alpha}}/\boldsymbol{\alpha}!|
     \leq \int_{0}^{1} (1-u)^{t-1} |g^{(t)}(u)-g^{(t)}(0)| du/(t-1)!, \label{reminder}
\end{equation}
and for any $u \in [0,1]$, with $f\in\mathcal H^\beta([0,1]^d,B)$, we have
\begin{align}
    |g^{(t)}(u)-g^{(t)}(0)| 
    &= \sum_{\|\boldsymbol{\alpha}\|_1 = t} t!h^{\boldsymbol{\alpha}}
    |\partial^{\boldsymbol{\alpha}} f(\boldsymbol{x}_0+u\boldsymbol{h})-\partial^{\boldsymbol{\alpha}} f(\boldsymbol{x}_0)|/\boldsymbol{\alpha}! \nonumber
    \\
    & \leq \sum_{\|\boldsymbol{\alpha}\|_1 = t} t!\boldsymbol{h}^{\boldsymbol{\alpha}}
    B \|u\boldsymbol{h}\|_{\infty}^{\sigma}/\boldsymbol{\alpha}!
     \leq Bd^t \|\boldsymbol{x}-\boldsymbol{x}_0\|_{\infty}^{\zeta}. \label{g-g}
\end{align}
Combining \eqref{reminder} with \eqref{g-g} yields the following local Taylor remainder bound: for any $\boldsymbol{x}_0 \in [0,1]^d$ and $\boldsymbol{x}$ in a neighborhood of $\boldsymbol{x}_0$
\begin{align*}
    | f(\boldsymbol{x}) - \sum_{\|\boldsymbol{\alpha}\|_1 \leq t} \partial^{\boldsymbol{\alpha}} f(\boldsymbol{x_0}) (\boldsymbol{x} - \boldsymbol{x_0})^{\boldsymbol{\alpha}}/\boldsymbol{\alpha}!|
    &\leq \int_{0}^{1} (1-u)^{t-1}  Bd^t \|\boldsymbol{x}-\boldsymbol{x}_0\|_{\infty}^{\zeta}du/(t-1)! 
    \\
    &=   Bd^t \|\boldsymbol{x}-\boldsymbol{x}_0\|_{\infty}^{\zeta}/t!
    \leq Bd^t \|\boldsymbol{x}-\boldsymbol{x}_0\|_{\infty}^{\zeta}.
\end{align*} 
Thus we complete the proof of Lemma \ref{Lemma 3}.
\end{proof}

We then give the proof of Theorem \ref{Theorem approx}, which establishes an upper bound for the approximation error. The proof consists of three steps: step 1 constructs a simplicial partition of $[0,1]^d$, step 2 introduces the corresponding barycentric partition of unity and step 3 combines local Taylor approximations to derive the global approximation and error bound.

\begin{proof} In step 1, we give a simplicial partition of $[0,1]^d$.
Fix an integer $N \ge 1$ and define the mesh size $1/N$. We use the Freudenthal triangulation \citep{freudenthal1942}
to form a partition of unity of the domain $[0,1]^d$. We then introduce the definition of simplex in $\mathbb{R}^d$.

For any points $\boldsymbol{v}_1,\boldsymbol{v}_2,\cdots,\boldsymbol{v}_m\in\mathbb{R}^d$, we can denote a convex hull $\mathrm{Conv}(\boldsymbol{v}_1,\boldsymbol{v}_2,\cdots,\boldsymbol{v}_m)$ as the set of all convex combination of points $\boldsymbol{v}_i$:
\begin{equation*}
    \mathrm{Conv}(\boldsymbol{v}_1,\boldsymbol{v}_2,\cdots,\boldsymbol{v}_m):=\Big\{\boldsymbol{x}=\sum_{i=1}^m\lambda_i\boldsymbol{v}_i\::\:\sum_{i=1}^m\lambda_i\equiv 1,\:\lambda_i\ge 0,\:\forall i \Big\}.
\end{equation*}
A set $S\subseteq \mathbb{R}^d$ is called a simplex if it is the closed convex hull of $d+1$ affinely independent points $\{\boldsymbol v_0,\boldsymbol{v}_1,\cdots,\boldsymbol{v}_d\}$, i.e.,
$S=\mathrm{Conv}\left(\boldsymbol{v}_0,\boldsymbol{v}_1,\cdots,\boldsymbol{v}_d\right)=\{\boldsymbol{x}\::\:\boldsymbol{x}=\sum_{i=0}^d\lambda_i(\boldsymbol{x})\boldsymbol{v}_i\}$,
these points $\boldsymbol{v}_i$ are the vertices of $S$ and these affine functions $\lambda_i(\boldsymbol{x})$ are the barycentric coordinate functions.

Let $V := \{i/N:i=0,1,\ldots,N\}^d$ be the set of grid vertices and let $ L:=\{i/N:i=0,1,\ldots,N-1\}^d$ be the set of lower left corners of the grid cubes. For any $\boldsymbol{v} \in L$ and any permutation $\pi$ of $\{1,2,\ldots,d\}$, we denote the simplices with size $1/N$ in $[0,1]^d$ as 
    $ S_{\boldsymbol{v},\pi} := \left\{\boldsymbol{x} \in [0,1]^d : 0 \leq x_{\pi(1)} - v_{\pi(1)} \leq \cdots \leq x_{\pi(d)} - v_{\pi(d)} \leq 1/N\right\}$.
    
The corresponding uniform simplicial mesh $\mathcal{S}_N$ with size $1/N$ on $[0,1]^d$ is  
    $ 
    \mathcal{S}_N := \{S_{\boldsymbol{v},\pi} : \boldsymbol{v} \in L, \pi \text{ is a permutation of } \{1,2,\ldots,d\}\}  
    $,
satisfying $\bigcup_{S_{\boldsymbol{v},\pi}\in\mathcal S_N} S_{\boldsymbol{v},\pi}= [0,1]^d$ and $S_{\boldsymbol{v},\pi}^\circ \cap (S_{\boldsymbol{v}',\pi'})^\circ = \varnothing$
for all $S_{\boldsymbol{v},\pi},S_{\boldsymbol{v}',\pi'} \in \mathcal S_N,\ S_{\boldsymbol{v},\pi} \neq S_{\boldsymbol{v}',\pi'}$, where $A^\circ$ denotes the interior of $A$. Moreover, each simplex has diameter of order $N^{-1}$:
    $    \operatorname{diam}(S_{\boldsymbol{v},\pi})
  := \sup_{\boldsymbol{x},\boldsymbol{y}\in S_{\boldsymbol{v},\pi}}\|\boldsymbol{x}-\boldsymbol{y}\|_\infty
  \le N^{-1},
  ~ \forall S_{\boldsymbol{v},\pi}\in\mathcal S_N.$

In step 2, we give the barycentric partition of unity. Let $\operatorname{Vert}(S) = \{\boldsymbol{v}_0,\dots,\boldsymbol{v}_d\}\subset V$ be a vertex set of a simplex $S$ in $\mathcal S_N$.
For any $\boldsymbol{x}\in S$ and $\boldsymbol{v}_i \in \operatorname{Vert}(S)$, $\boldsymbol{x}$ admits a unique barycentric representation: $ \boldsymbol{x} = \sum_{i=0}^{d} \psi_{\boldsymbol{v}_i}(\boldsymbol{x})\, {\boldsymbol{v}}_i$, $0 \leq \lambda_i(\boldsymbol{x}) \leq 1$, $\sum_{i=0}^{d}\lambda_i(\boldsymbol{x}) \equiv 1$.

For each grid vertex $\boldsymbol{v}_i\in V$, we define 
$\psi_{\boldsymbol{v}_i}(\cdot):[0,1]^d \to [0,1]$ as follows:
\begin{align*}
  \psi_{\boldsymbol{v}_i}(\boldsymbol{x})
  :=
  \begin{cases}
    \lambda_i(\boldsymbol{x}),
      & \text{if } \boldsymbol{x}\in S\in\mathcal S_N
        \text{ and } \boldsymbol{v}_i \in \operatorname{Vert}(S),\\[3pt]
    0,
      & \text{otherwise},
  \end{cases}
\end{align*}
where $S$ may be chosen as any simplex in $\mathcal S_N$ containing $x$. 
This definition is independent of the choice of $S$, since barycentric coordinates agree on common faces of adjacent simplices. Hence $\psi_{\boldsymbol{v}_i}(\cdot)$ is globally well-defined and continuous on $[0,1]^d$. 

For a fixed $\boldsymbol{x}$, only $\psi_{\boldsymbol{v}_i}(\boldsymbol{x})$ corresponding to the $d+1$ vertices of the simplex containing $\boldsymbol{x}$ are non-zero, while $\psi_{\boldsymbol{v}_i}(\boldsymbol{x})$ are zero for all other grid points $\boldsymbol{v}_i$. Then for any $\boldsymbol{x}\in[0,1]^d$, let $S\in\mathcal S_N$ denote the
simplex containing $\boldsymbol{x}$, we obtain that
\begin{equation*}
  \sum_{\boldsymbol{v}_i\in V} \psi_{\boldsymbol{v}_i}(\boldsymbol{x}) 
  = \sum_{\boldsymbol{v}_i\in \operatorname{Vert}(S)} \psi_{\boldsymbol{v}_i}(\boldsymbol{x}) 
  = \sum_{i=0}^{d} \psi_{\boldsymbol{v}_i}(\boldsymbol{x})
  \equiv 1,
\end{equation*}
thus the collection $\{\psi_{\boldsymbol{v}} : \boldsymbol{v}\in V\}$ forms a global partition of unity $\sum_{\boldsymbol{v}\in V} \psi_{\boldsymbol{v}}(\boldsymbol{x}) \equiv 1$ subordinate to the simplicial mesh $\mathcal S_N$.

In step 3, we give a local taylor approximation.
For each $\boldsymbol{v}\in V$, we define the $t$-th order Taylor polynomial of $f$ at $\boldsymbol{x} = \boldsymbol{v}$ by
\begin{align*}
  T_{\boldsymbol{v}}(\boldsymbol{x})
  :=
  \sum_{\|\boldsymbol{\alpha}\|_1\le t}
    \tau_{\boldsymbol{v},\boldsymbol{\alpha}}\,(\boldsymbol{x}-\boldsymbol{v})^{\boldsymbol{\alpha}},
  \quad
  \tau_{\boldsymbol{v},\boldsymbol{\alpha}}
  :=
  \partial^{\boldsymbol{\alpha}} f(\boldsymbol{v})/\boldsymbol{\alpha}!,
  \quad
  \boldsymbol{\alpha}:= (\alpha_1,\alpha_2,\cdots,\alpha_d) \in \mathbb{N}^d,
\end{align*}
where $t=\lfloor  \zeta \rfloor $, $\boldsymbol{\alpha}!=\alpha_{1}!\cdots \alpha_{d}!$. 
Using the barycentric partition of unity $\{\psi_{\boldsymbol{v}}:\boldsymbol{v}\in V\}$, we define
  $p(\boldsymbol{x})
  :=
  \sum_{\boldsymbol{v}\in V}
  \psi_{\boldsymbol{v}}(\boldsymbol{x})T_{\boldsymbol{v}}(\boldsymbol{x}).$
  \label{eq:global-p}

For a fixed point $\boldsymbol{x} \in [0,1]^d$ contained in a simplex $S$ and $f\in\mathcal H^\beta([0,1]^d,B)$, using Lemma \ref{Lemma 3}, we have 
\begin{align}
    |f(\boldsymbol{x}) - p(\boldsymbol{x})|
  = &
  |\sum_{\boldsymbol{v}\in V}
  \psi_{\boldsymbol{v}}(\boldsymbol{x})f(\boldsymbol{x})-\sum_{\boldsymbol{v}\in V}
  \psi_{\boldsymbol{v}}(\boldsymbol{x})T_{\boldsymbol{v}}(\boldsymbol{x})| \nonumber 
  \leq 
  \sum_{\boldsymbol{v}\in V}
  \psi_{\boldsymbol{v}}(\boldsymbol{x})|(f(\boldsymbol{x})-T_{\boldsymbol{v}}(\boldsymbol{x})| 
  \\
  \leq&
  \sum_{\boldsymbol{v}\in \operatorname{Vert}(S)}|f(\boldsymbol{x})-\sum_{\|\boldsymbol{\alpha}\|_1\le t}
  \tau_{\boldsymbol{v},\boldsymbol{\alpha}}\,(\boldsymbol{x}-\boldsymbol{v})^{\boldsymbol{\alpha}}|\nonumber
  \leq
  (d+1)Bd^t\|\boldsymbol{x}-\boldsymbol{v}\|_\infty^{\,\zeta} 
  \\
  \leq&
  (d+1)Bd^tN^{-\zeta}.
  \label{eq:|f-p|}
\end{align}

Compared with the error bound in \citet{feng2024deep}, \eqref{eq:|f-p|} replaces the exponential overlap factor $2^d$ by the polynomial term $d+1$, yielding a dimensionally much more favorable approximation constant.
Using similar proof in \citet{feng2024deep}, there exists a neural network $\hat{f}(\boldsymbol{x})=\sum_{\boldsymbol{v}\in V}\sum_{\substack{\|\boldsymbol{\alpha}\|_1 \le t}} \tau_{\boldsymbol{v},\boldsymbol{\alpha}}\tilde f_{\boldsymbol{v},\boldsymbol{\alpha}}(\boldsymbol{x})$ that can be an approximation of $p(\boldsymbol{x})$:
\begin{align}
|p(\boldsymbol{x})-\hat{f}(\boldsymbol{x})| & =|\sum_{\boldsymbol{v}\in V}
  \psi_{\boldsymbol{v}}(\boldsymbol{x})\sum_{\|\boldsymbol{\alpha}\|_1\le t}
    \tau_{\boldsymbol{v},\boldsymbol{\alpha}}\,(\boldsymbol{x}-\boldsymbol{v})^{\boldsymbol{\alpha}}-\sum_{\boldsymbol{v}\in V}\sum_{\substack{\|\boldsymbol{\alpha}\|_1 \le t}} \tau_{\boldsymbol{v},\boldsymbol{\alpha}}\tilde f_{\boldsymbol{v},\boldsymbol{\alpha}}(\boldsymbol{x})| \nonumber\\
 & \leq \sum_{\boldsymbol{v}\in V}\sum_{\substack{\|\boldsymbol{\alpha}\|_1 \le t}} |\tau_{\boldsymbol{v},\boldsymbol{\alpha}}||\psi_{\boldsymbol{v}}(\boldsymbol{x}) (\boldsymbol{x}-\boldsymbol{v})^{\boldsymbol{\alpha}}- \tilde f_{\boldsymbol{v},\boldsymbol{\alpha}}(\boldsymbol{x})| \nonumber\\
 & \leq \sum_{\boldsymbol{v}\in V} (t+1)d^tB\max_{\|\boldsymbol{\alpha}\|_1 \le t}|\psi_{\boldsymbol{v}}(\boldsymbol{x}) (\boldsymbol{x}-\boldsymbol{v})^{\boldsymbol{\alpha}}- \tilde f_{\boldsymbol{v},\boldsymbol{\alpha}}(\boldsymbol{x})| \nonumber\\
 & \leq (t+1)d^tB \sum_{\boldsymbol{v}\in \operatorname{Vert}(S)}\max_{\|\boldsymbol{\alpha}\|_1 \le t}|\psi_{\boldsymbol{v}}(\boldsymbol{x}) (\boldsymbol{x}-\boldsymbol{v})^{\boldsymbol{\alpha}}- \tilde f_{\boldsymbol{v},\boldsymbol{\alpha}}(\boldsymbol{x})| \nonumber\\
 & \leq (t+1)d^tB(d+1) \max_{\boldsymbol{v}\in \operatorname{Vert}(S)}\max_{\|\boldsymbol{\alpha}\|_1 \le t} |\psi_{\boldsymbol{v}}(\boldsymbol{x}) (\boldsymbol{x}-\boldsymbol{v})^{\boldsymbol{\alpha}}- \tilde f_{\boldsymbol{v},\boldsymbol{\alpha}}(\boldsymbol{x})|,\nonumber \label{(6)}
\end{align}
where the last inequality follows from the fact that a $d$-dimensional simplex has $d+1$ vertices. Denoted by $p_{\boldsymbol{v},\boldsymbol{\alpha}}(\boldsymbol{x})
:= \psi_{\boldsymbol{v}}(\boldsymbol{x})\,(\boldsymbol{x}-\boldsymbol{v})^{\boldsymbol{\alpha}}$, using the similar construction and results of \citet{yarotsky2017error} and \citet{feng2024deep}(lemma 18), we can define the neural network $\tilde f_{\boldsymbol{v},\boldsymbol{\alpha}}(\boldsymbol{x})$ with the depth
and size not larger than $c_1(d+t)\ln{(1/\delta)}$ for some constants $c_1=c_1(d,t)$, satisfying:
    $|\psi_{\boldsymbol{v}}(\boldsymbol{x}) (\boldsymbol{x}-\boldsymbol{v})^{\boldsymbol{\alpha}}- \tilde f_{\boldsymbol{v},\boldsymbol{\alpha}}(\boldsymbol{x})| \leq (d+t)\delta,$
hence we obtain $|p(\boldsymbol{x})-\hat{f}(\boldsymbol{x})|\leq (t+1)d^tB(d+1)(d+t)\delta =(t+1)(d+1)(d+t)Bd^t\delta$.

These show that the constants in front of $N^{-\zeta}$ and $\delta$
depend on the dimension $d$ only polynomially, as $(d+1)(d+t) d^t$, rather
than exponentially as $2^d d^t$ in the \citet{feng2024deep}. Consequently, to reach a target accuracy $\varepsilon>0$, it suffices to choose
\begin{equation}
    N =\lceil \varepsilon^{-1/\zeta}(2(d+1)Bd^t)^{1/\zeta}\rceil
    \quad\text{and}\quad
    \delta = \varepsilon (2(t+1)(d+1)(d+t)Bd^t)^{-1},
    \label{eq:N,delte}
\end{equation}
so that the required network complexity grows only polynomially in $d$. Therefore, we have 
\begin{center}
    $|f(\boldsymbol{x})-\hat{f}(\boldsymbol{x})|=|f(\boldsymbol{x})-p(\boldsymbol{x})+p(\boldsymbol{x})-\hat{f}(\boldsymbol{x})| \leq |f(\boldsymbol{x})-p(\boldsymbol{x})| + |p(\boldsymbol{x})-\hat{f}(\boldsymbol{x})| \leq \varepsilon$.
\end{center}
Thus we complete the proof of Theorem \ref{Theorem approx}
\end{proof}

Here we illustrate the parameters of the ReLU FNNs $\hat{f}=\sum_{\boldsymbol{v}\in V}\sum_{\substack{\|\boldsymbol{\alpha}\|_1 \le t}} \tau_{\boldsymbol{v},\boldsymbol{\alpha}}\tilde f_{\boldsymbol{v},\boldsymbol{\alpha}}(\boldsymbol{x})$. $\hat{f}$ is the linear combination of $d^t(N+1)^d$ local networks. According to the lemma 19 of \citet{feng2024deep}, the depth and size of $\hat{f}$ are $c_2\ln(1/\delta) + 1$ and $d^t(N+1)^d(c_2\ln(1/\delta) + 1)$, respectively, $c_2 = c_2(d,t)$. For $N$ and $\delta$ in (\ref{eq:N,delte}), then $\hat{f}$ has the depth at most $c_2\ln(1/\varepsilon)+c_3$, the size at most $\varepsilon ^{-d/\zeta}(c_2\ln(1/\varepsilon)+c_3) + c_4$, $c_3 = c_3(d,t)$, $c_4=c_4(d,t,\zeta)$ and the weights bound at most $c_5\varepsilon ^{-d/\zeta}$, $c_5=c_5(d,t,\zeta)$. All the constants which depend on the dimension $d$ only polynomially (through $(d+1)(d+t)d^t$).

\subsubsection{Proof of Theorem \ref{Theorem 3rb}}

\begin{proof}
We first decompose the mean squared error. Reviewing the previous definition of $L(\cdot)$ and $\widehat{L}(\cdot)$, for each $v \in \mathcal{V}$
\begin{align*}
    & L(v):=\mathbb{E}_{\boldsymbol{X}\sim P_{\boldsymbol{X}}}(v(\boldsymbol{X})^2)/2-\mathbb{E}_{\boldsymbol{X}\sim Q_{\boldsymbol{X}}}(v(\boldsymbol{X})) ,
    \\&
    \widehat{L}(v):=\sum_{i=n_1^{P}+1}^{n^{P}}\sum_{j=1}^{m_{i}^{P}}v(\boldsymbol{X}_{ij}^{P})^{2}/(2N_{2}^{P})-\sum_{i=1}^{n^{Q}}\sum_{j=1}^{m_{i}^{Q}}v(\boldsymbol{X}_{ij}^{Q})/N^{Q}.
\end{align*}
Through straightforward calculations and the definition of $\hat{r}(\cdot)$ in \eqref{r_hat}, we obtain 
\begin{align*}
    &\mathbb{E}_{\boldsymbol{X}\sim P_{\boldsymbol{X}}}(\hat{r}(\boldsymbol{X})-r(\boldsymbol{X}))^2 =\mathbb{E}_{\boldsymbol{X}\sim P_{\boldsymbol{X}}}(\hat{r}(\boldsymbol{X})-r(\boldsymbol{X}))^2-\mathbb{E}_{\boldsymbol{X}\sim P_{\boldsymbol{X}}}(r(\boldsymbol{X})-r(\boldsymbol{X}))^2\\
    =&\mathbb{E}_{\boldsymbol{X}\sim P_{\boldsymbol{X}}}\hat{r}(\boldsymbol{X})^2-2\mathbb{E}_{\boldsymbol{X}\sim Q_{\boldsymbol{X}}}\hat{r}(\boldsymbol{X})-\mathbb{E}_{\boldsymbol{X}\sim P_{\boldsymbol{X}}}r(\boldsymbol{X})^2+2\mathbb{E}_{\boldsymbol{X}\sim Q_{\boldsymbol{X}}}r(\boldsymbol{X})\\
    =& 2L(\hat{r})-2L(r) =2L(r)-4\widehat{L}(\hat{r})+2L(\hat{r})+4\widehat{L}(\hat{r})-4L(r) \\
     \leq &2L(r)-4\widehat{L}(\hat{r})+2L(\hat{r})+4\widehat{L}(v)-4L(r).
\end{align*}
Taking expectations on both sides followed by taking the infimum over $v\in \mathcal{V}$ yields
\begin{align}
    \mathbb{E}[\|\hat{r}-r\|_{L_2(P_{\boldsymbol{X}})}^2]&\leq\inf_{v\in \mathcal{V}}\mathbb{E}[2L(r)-4\widehat{L}(\hat{r})+2L(\hat{r})+4\widehat{L}(v)-4L(r)]\nonumber\\
    & = \mathbb{E}[2L(r)-4\widehat{L}(\hat{r})+2L(\hat{r})]+\inf_{v\in \mathcal{V}}\mathbb{E}[4\widehat{L}(v)-4L(r)]\nonumber\\
    & :=(I)+(II).
    \label{T3(I)(II)}
\end{align}
Consequently, the risk can be decomposed into a stochastic error term (I) and an approximation error term (II).
To control the stochastic error term (I), we utilize the result of Theorem \ref{Theorem rstochastic}. Setting the depth $\mathcal{D}=O(\log N_r)$, the weights $\mathcal{B}=O((N_r)^{d/(2\alpha+d)})$ and the size $\mathcal{S}=O((N_r)^{d/(2\alpha+d)}\log N_r)$, then we obtain
\begin{align}
    (I)\lesssim \Gamma(\Gamma+1)\mathcal{DS}\log\mathcal{S}\log N_r/N_r\lesssim\Gamma(\Gamma+1)(N_r)^{-2\alpha/(d+2\alpha)}(\log N_r)^5.
    \label{T3(I)result}
\end{align}
To control the approximation error term (II), we first simplify it as
\begin{align*}
    &(II)=\inf_{v\in \mathcal{V}}\left(4L(v)-4L(r)\right)\nonumber =\inf_{v\in \mathcal{V}}(2\mathbb{E}_{\boldsymbol{X}\sim P_{\boldsymbol{X}}}(v(\boldsymbol{X})^2-r(\boldsymbol{X})^2)-4\mathbb{E}_{\boldsymbol{X}\sim Q_{\boldsymbol{X}}}(v(\boldsymbol{X})-r(\boldsymbol{X})))\nonumber\\
    & =\inf_{v\in \mathcal{V}}(2\mathbb{E}_{\boldsymbol{X}\sim P_{\boldsymbol{X}}}(v(\boldsymbol{X})-r(\boldsymbol{X}))^2+4\mathbb{E}_{\boldsymbol{X}\sim P_{\boldsymbol{X}}}r(\boldsymbol{X})(v(\boldsymbol{X})-r(\boldsymbol{X}))-4\mathbb{E}_{\boldsymbol{X}\sim Q_{\boldsymbol{X}}}(v(\boldsymbol{X})-r(\boldsymbol{X})))\nonumber\\
    & =2\inf_{v\in \mathcal{V}}\|v-r\|_{L_2(P_{\mathbf{X}})}^2.
\end{align*}
By using the result of Theorem \ref{Theorem approx}, when we choose $\varepsilon=(N_r)^{-\alpha/(2\alpha+d)}$, the approximation error satisfies
\begin{align}
 (II) \lesssim N_r^{-{2\alpha}/(d+2\alpha)}\log N_r.
    \label{T3(II)result}
\end{align}
Finally, substituting \eqref{T3(I)result} and \eqref{T3(II)result} into \eqref{T3(I)(II)}, we complete the proof.
\end{proof}

\subsubsection{Proof of Theorem \ref{Theorem r_xi_hat}}

\begin{proof}
We first decompose the mean squared error. Through straightforward calculations, for each $v\in T_{\xi}\mathcal{V}$, we similarly obtain that
\begin{align*}
    \mathbb{E}_{\boldsymbol{X}\sim P_{\boldsymbol{X}}}(\hat{r}_{\xi}(\boldsymbol{X})-r(\boldsymbol{X}))^2& =\mathbb{E}_{\boldsymbol{X}\sim P_{\boldsymbol{X}}}(\hat{r}_{\xi}(\boldsymbol{X})-r(\boldsymbol{X}))^2-\mathbb{E}_{\boldsymbol{X}\sim P_{\boldsymbol{X}}}(r(\boldsymbol{X})-r(\boldsymbol{X}))^2\\
    & =\mathbb{E}_{\boldsymbol{X}\sim P_{\boldsymbol{X}}}\hat{r}_{\xi}(\boldsymbol{X})^2-2\mathbb{E}_{\boldsymbol{X}\sim Q_{\boldsymbol{X}}}\hat{r}_{\xi}(\boldsymbol{X})-\mathbb{E}_{\boldsymbol{X}\sim P_{\boldsymbol{X}}}r(\boldsymbol{X})^2+2\mathbb{E}_{\boldsymbol{X}\sim Q_{\boldsymbol{X}}}r(\boldsymbol{X})\\
    & =2L(\hat{r}_{\xi})-2L(r)\\
    & =2L(r)-4\widehat{L}(\hat{r}_{\xi})+2L(\hat{r}_{\xi})+4\widehat{L}(\hat{r}_{\xi})-4L(r) \\
    & \leq 2L(r)-4\widehat{L}(\hat{r}_{\xi})+2L(\hat{r}_{\xi})+4\widehat{L}(v)-4L(r).
\end{align*}
Taking expectations on both sides followed by taking the infimum over $v\in T_{\xi}\mathcal{V}$ yields
\begin{align}
    \mathbb{E}[\|\hat{r}_{\xi}-r\|_{L_2(P_{\boldsymbol{X}})}^2]&\leq\inf_{v\in T_\xi\mathcal{V}}\mathbb{E}[2L(r)-4\widehat{L}(\hat{r}_{\xi})+2L(\hat{r}_{\xi})+4\widehat{L}(v)-4L(r)]\nonumber\\
    & = \mathbb{E}[2L(r)-4\widehat{L}(\hat{r}_{\xi})+2L(\hat{r}_{\xi})]+\inf_{v\in T_\xi\mathcal{V}}\mathbb{E}[4\widehat{L}(v)-4L(r)]\nonumber\\
    & :=(I)+(II).
    \label{T5(I)(II)}
\end{align}
Thus, the risk can be decomposed into a stochastic error term (I) and an approximation error term (II). 
By the result of \eqref{T3(I)result}, we obtain that (I) can be bounded by $O(\xi^2\mathcal{DS}\log \mathcal{S}\log N_r/{N_r})$. To control the approximation error term (II), we first simplify it as
\begin{align}
    &(II)=\inf_{v\in T_\xi\mathcal{V}}(4L(v)-4L(r))\nonumber\\
    & =\inf_{v\in T_\xi\mathcal{V}}2[\mathbb{E}_{\boldsymbol{X}\sim P_{\boldsymbol{X}}}(v(\boldsymbol{X})^2-r(\boldsymbol{X})^2)-2\mathbb{E}_{\boldsymbol{X}\sim Q_{\boldsymbol{X}}}(v(\boldsymbol{X})-r(\boldsymbol{X}))]\nonumber\\
    & =\inf_{v\in T_\xi\mathcal{V}}2[\mathbb{E}_{\boldsymbol{X}\sim P_{\boldsymbol{X}}}(v(\boldsymbol{X})-r(\boldsymbol{X}))^2+2\mathbb{E}_{\boldsymbol{X}\sim P_{\boldsymbol{X}}}r(\boldsymbol{X})(v(\boldsymbol{X})-r(\boldsymbol{X}))-2\mathbb{E}_{\boldsymbol{X}\sim Q_{\boldsymbol{X}}}(v(\boldsymbol{X})-r(\boldsymbol{X}))]\nonumber\\
    & =\inf_{v\in T_\xi\mathcal{V}}2\|v-T_{\xi}r+T_{\xi}r-r\|_{L_2(P_{\mathbf{X}})}^2\nonumber\\
    & \leq4\inf_{v\in T_\xi\mathcal{V}}\|v-T_{\xi}r\|_{L_2(P_{\mathbf{X}})}^2+4\|T_{\xi}r-r\|_{L_2(P_{\mathbf{X}})}^2:=(A)+(B).
    \label{T5(II)(A)(B)}
\end{align}
We can control the first term (A) by Theorem \ref{Theorem approx}. Then using the Assumption \ref{assumption rub} and Markov inequality, we obtain
\begin{align}
    (B)& =4\mathbb{E}_{\boldsymbol{X}\sim P_{\boldsymbol{X}}}(r(\boldsymbol{X})-\xi)^2I(r\geq\xi)\nonumber\\
    & \leq8\mathbb{E}_{\boldsymbol{X}\sim P_{\boldsymbol{X}}}r(\boldsymbol{X})^2I(r\geq\xi)+8\mathbb{E}_{\boldsymbol{X}\sim P_{\boldsymbol{X}}}\xi^2I(r\geq\xi)\nonumber\\
    & \leq8\mathbb{E}_{\boldsymbol{X}\sim P_{\boldsymbol{X}}}[r(\boldsymbol{X})^2r(\boldsymbol{X})^{\delta}/{\xi^{\delta}}]+8\xi^2\mathbb{E}_{\boldsymbol{X}\sim P_{\boldsymbol{X}}}[r(\boldsymbol{X})^{\delta+2}/{\xi^{\delta+2}}]\nonumber\\
    & =16\kappa_{\delta}/{\xi^{\delta}}.
    \label{T5(IIB)}
\end{align}
Combining the result of \eqref{T5(I)(II)}, \eqref{T5(II)(A)(B)} and \eqref{T5(IIB)}, we have
\begin{center}
    $\mathbb{E}[\|\hat{r}_{\xi}-r\|_{L_2(P_{\boldsymbol{X}})}^2]\lesssim\xi^2\mathcal{DS}\log \mathcal{S}\log N_r/{N_r}+16\kappa_{\delta}/{\xi^{\delta}}$,
\end{center}
then setting $\xi\asymp(\kappa_{\delta}N_r/\mathcal{DS}\log \mathcal{S}\log N_r)^{1/(2+\delta)}$, the size $\mathcal{S}=O((N_r)^{d/(d+(4/\delta+2)\alpha)}\log N_r)$, the depth $\mathcal{D}=O(\log N_r)$, and weights bound $\mathcal{B}=O((N_r)^{d/(d+(4/\delta+2)\alpha)})$, it follows that
\begin{center}
    $\mathbb{E}[\|\hat{r}_{\xi}-r\|_{L_2(P_{\boldsymbol{X}})}^2]\lesssim \kappa_{\delta}^{2/(2+\delta)}(N_r)^{-2\alpha/(d+(2+4/\delta)\alpha)}(\log N_r)^5$.
\end{center}
Thus we complete the proof.
\end{proof}

\subsection{Proofs for Target Regression Estimators with an Unknown Density Ratio}\label{ProofURE}
In this section, we give the proofs of Theorems \ref{Theorem Ef^_0_r^} and \ref{Theorem Ef^_0_r^_zeta}, which establish the non-asymptotic error upper bounds for the target regression estimator with an unknown density ratio under bounded and unbounded density ratio settings, respectively.

\subsubsection{Proof of Theorem \ref{Theorem Ef^_0_r^}}

\begin{proof}
First, we similarly decompose the error term. We define $R(f)=\mathbb{E}_{(\boldsymbol{X},Y)\thicksim P_{\boldsymbol{X},Y}}[r(\boldsymbol{X})\\(Y-f(\boldsymbol{X}))^2]$ for any measurable function $f(\cdot)$ and $f^*(\cdot)\in\underset{f\in\mathcal{F}}{\operatorname*{\mathrm{argmin}}}\|f-f_0\|_{L_2(P_X)}^2$. Then we obtain
\begin{align*}
R(f)  =&\mathbb{E}_{(\boldsymbol{X},Y)\thicksim P_{\boldsymbol{X},Y}}[r(\boldsymbol{X})(Y-f_{0}(\boldsymbol{X})+f_{0}(\boldsymbol{X})-f(\boldsymbol{X}))^{2}] \\
  =&\mathbb{E}_{(\boldsymbol{X},Y)\thicksim P_{\boldsymbol{X},Y}}[r(\boldsymbol{X})(Y-f_0(\boldsymbol{X}))^2]+\mathbb{E}_{\boldsymbol{X}\thicksim P_{\boldsymbol{X}}}[r(\boldsymbol{X})(f_0(\boldsymbol{X})-f(\boldsymbol{X}))^2]\\
 &+2\mathbb{E}_{(\boldsymbol{X},Y)\thicksim P_{\boldsymbol{X},Y}}[r(\boldsymbol{X})(Y-f_0(\boldsymbol{X}))(f_0(\boldsymbol{X})-f(\boldsymbol{X}))] \\
  =&\mathbb{E}_{\boldsymbol{X}\thicksim P_{\boldsymbol{X}}}[r(\boldsymbol{X})Var(Y|\boldsymbol{X}))]+\mathbb{E}_{\boldsymbol{X}\thicksim P_{\boldsymbol{X}}}[r(\boldsymbol{X})(f_{0}(\boldsymbol{X})-f(\boldsymbol{X}))^{2}].
\end{align*}
Through straightforward calculations, we have
\begin{align}
    \mathbb{E}[\|\hat{f}_{0,\hat{r}}-f_0\|_{L_2(Q_{\boldsymbol{X}})}^2]\nonumber
     &=\mathbb{E}[R(\hat{f}_{0,\hat{r}})-R(f_0)]\nonumber\\
     &=\mathbb{E}[\mathbb{E}_{(\boldsymbol{X},Y)\thicksim P_{\boldsymbol{X},Y}}[(r(\boldsymbol{X})-\hat{r}(\boldsymbol{X}))((Y-\hat{f}_{0,\hat{r}}(\boldsymbol{X}))^2-(Y-f_0(\boldsymbol{X}))^2)]]\nonumber\\
    &~~~+\mathbb{E}[\mathbb{E}_{(\boldsymbol{X},Y)\thicksim P_{\boldsymbol{X},Y}}[\hat{r}(\boldsymbol{X})((Y-\hat{f}_{0,\hat{r}}(\boldsymbol{X}))^2-(Y-f_0(\boldsymbol{X}))^2)]]\nonumber\\
     &:=(I)+(II).
    \label{T4(I)(II)}
\end{align}
To control term (I), we aplly the GM-QM inequality. Under the Assumptions \ref{assumption rb0}, \ref{assumption f0B1} and \ref{assumption fiepsilon}, for any $\epsilon_{0}>0$, we have
\begin{align*}
    &\mathbb{E}[\mathbb{E}_{(\boldsymbol{X},Y)\thicksim P_{\boldsymbol{X},Y}}[(r(\boldsymbol{X})-\hat{r}(\boldsymbol{X}))((Y-\hat{f}_{0,\hat{r}}(\boldsymbol{X}))^2-(Y-f_0(\boldsymbol{X}))^2)]]\\
     \leq&\epsilon_{0}\mathbb{E}[\mathbb{E}_{(\boldsymbol{X},Y)\thicksim P_{\boldsymbol{X},Y}}[r(\boldsymbol{X})((Y-\hat{f}_{0,\hat{r}}(\boldsymbol{X}))^2-(Y-f_0(\boldsymbol{X}))^2)^2]]/4\\
    \quad&+\mathbb{E}[\mathbb{E}_{\boldsymbol{X}\thicksim P_{\boldsymbol{X}}}[(r(\boldsymbol{X})-\hat{r}(\boldsymbol{X}))^2/r(\boldsymbol{X})]]/\epsilon_{0}\\
     \leq&\epsilon_{0}\mathbb{E}[\mathbb{E}_{(\boldsymbol{X},Y)\thicksim P_{\boldsymbol{X},Y}}[r(\boldsymbol{X})(f_0(\boldsymbol{X})-\hat{f}_{0,\hat{r}}(\boldsymbol{X}))^2(2Y-\hat{f}_{0,\hat{r}}(\boldsymbol{X})-f_0(\boldsymbol{X}))^2]]/4\\
     \quad&+\mathbb{E}[\|r-\hat{r}\|_{L_2(P_{\boldsymbol{X}})}^2]/(\epsilon_{0} M)\\
     \leq&4\epsilon_{0}(B_1^2+B_2^2+B_3^2)\mathbb{E}[\|\hat{f}_{\hat{r},\mathbb{D}}-f_0\|_{L_2(Q_{\boldsymbol{X}})}^2]+\mathbb{E}[\|r-\hat{r}\|_{L_2(P_{\boldsymbol{X}})}^2]/(\epsilon_{0} M).
\end{align*}
Then, setting $\epsilon_{0}^{-1}=8(B_1^2+B_2^2+B_3^2)$, we obtain
\begin{align}
    (I)\leq8(B_1^2+B_2^2+B_3^2)\mathbb{E}[\|r-\hat{r}\|_{L_2(P_{\boldsymbol{X}})}^2]/M+\mathbb{E}[\|\hat{f}_{0,\hat{r}}-f_0\|_{L_2(Q_{\boldsymbol{X}})}^2]/2.
    \label{T4(I)result}
\end{align}
Next, we bound the term (II). For each $f(\cdot)\in \mathcal{F}$, we define
\begin{align*}
    & L_{\hat{r}}(f)=\mathbb{E}_{\boldsymbol{X},Y\sim P_{\boldsymbol{X},Y}}[\hat{r}(\boldsymbol{X})((Y-f(\boldsymbol{X}))^2-(Y-f_0(\boldsymbol{X}))^2)],\\
    & \widehat{L}_{\hat{r}}(f)=\sum_{i=1}^{n_1^{P}}\sum_{j=1}^{m^{P}_{i}}
   \hat{r}(\boldsymbol{X}_{ij}^{P})((Y_{ij}^{P} - f(\boldsymbol{X}_{ij}^{P}))^2-(Y_{ij}^{P} - f_0(\boldsymbol{X}_{ij}^{P}))^2)/N_1^{P}.
\end{align*}
Through straightforward calculations and the definition of $\hat{f}_{0,\hat{r}}$ in \eqref{f_hat_r_hat}, for any $f(\cdot)\in \mathcal{F}$, we obtain
$L_{\hat{r}}(\hat{f}_{0,\hat{r}})
  =L_{\hat{r}}(\hat{f}_{0,\hat{r}})-2\widehat{L}_{\hat{r}}(\hat{f}_{0,\hat{r}})+2\widehat{L}_{\hat{r}}(\hat{f}_{0,\hat{r}})
 \leq L_{\hat{r}}(\hat{f}_{0,\hat{r}})-2\widehat{L}_{\hat{r}}(\hat{f}_{0,\hat{r}})+2\widehat{L}_{\hat{r}}(f),$
taking the infimum over $f\in \mathcal{F}$ followed by taking expectations on both sides yields
\begin{align}
    \mathbb{E}[L_{\hat{r}}(\hat{f}_{0,\hat{r}})]\leq&\mathbb{E}[L_{\hat{r}}(\hat{f}_{0,\hat{r}})-2\widehat{L}_{\hat{r}}(\hat{f}_{0,\hat{r}})]+\mathbb{E}[\inf_{f\in \mathcal{F}}2\widehat{L}_{\hat{r}}(f)]\nonumber\\
     \leq&\mathbb{E}[L_{\hat{r}}(\hat{f}_{0,\hat{r}})-2\widehat{L}_{\hat{r}}(\hat{f}_{0,\hat{r}})]+\mathbb{E}[2\widehat{L}_{\hat{r}}(f^*)]\nonumber\\
     =&\mathbb{E}[L_{\hat{r}}(\hat{f}_{0,\hat{r}})-2\widehat{L}_{\hat{r}}(\hat{f}_{0,\hat{r}})]+2L_{\hat{r}}(f^*)\nonumber\\
     =&\mathbb{E}[L_{\hat{r}}(\hat{f}_{0,\hat{r}})-2\widehat{L}_{\hat{r}}(\hat{f}_{0,\hat{r}})]+2\mathbb{E}_{\boldsymbol{X},Y\sim P_{\boldsymbol{X},Y}}[\hat{r}(\boldsymbol{X})((Y-f^*(\boldsymbol{X}))^2-(Y-f_0(\boldsymbol{X}))^2)]\nonumber\\
     =&\mathbb{E}[L_{\hat{r}}(\hat{f}_{0,\hat{r}})-2\widehat{L}_{\hat{r}}(\hat{f}_{0,\hat{r}})]+2\mathbb{E}_{\boldsymbol{X},Y\sim P_{\boldsymbol{X},Y}}[\hat{r}(\boldsymbol{X})(f_0(\boldsymbol{X})-f^*(\boldsymbol{X}))^2]\nonumber\\
    & +4\mathbb{E}_{\boldsymbol{X},Y\sim P_{\boldsymbol{X},Y}}[\hat{r}(\boldsymbol{X})(f_i(\boldsymbol{X})+\varepsilon)]\nonumber\\
     \leq&\mathbb{E}[L_{\hat{r}}(\hat{f}_{0,\hat{r}})-2\widehat{L}_{\hat{r}}(\hat{f}_{0,\hat{r}})]+2\Gamma\inf_{f\in \mathcal{F}}\|f-f_0\|_{L_2(P_{\mathbf{X}})}^2.
     \label{T4(II)sta}
\end{align}
Let $Y_{ij}^{\varepsilon,P}=f_0(\boldsymbol{X}_{ij}^P)+\varepsilon_{ij}$, then we define 
\begin{align*}
 & L^{\varepsilon}_{\hat{r}}(f)=\mathbb{E}[\hat{r}(\boldsymbol{X}_{ij}^P)((Y_{ij}^{\varepsilon,P}-f(\boldsymbol{X}_{ij}^P))^{2}-(Y^{\varepsilon,P}_{ij}-f_0(\boldsymbol{X}_{ij}^P))^{2})],\\
 & \hat{L}_{\hat{r}}^{\varepsilon}(f)=\sum_{i=1}^{n_1^{P}}\sum_{j=1}^{m^{P}_{i}}[\hat{r}(\boldsymbol{X}_{ij}^P)((Y_{i j}^{\varepsilon,P}-f(\boldsymbol{X}_{ij}^P))^{2}-(Y_{i j}^{\varepsilon,P}-f_0(\boldsymbol{X}_{ij}^P))^{2})]/N_1^{P},\\
 & L^{g}_{\hat{r}}(f)=\mathbb{E}_{\boldsymbol{X}\sim P_{\boldsymbol{X}}}[2\hat{r}(\boldsymbol{X})f_{i}(\boldsymbol{X})(f_0(\boldsymbol{X})-f(\boldsymbol{X}))],\\
 & \widehat{L}_{\hat{r}}^{g}(f)=2\sum_{i=1}^{n_1^{P}}\sum_{j=1}^{m^{P}_{i}}\hat{r}(\boldsymbol{X}_{i j}^P)f_{i}(\boldsymbol{X}_{ij}^P)(f_0(\boldsymbol{X}_{ij}^P)-f(\boldsymbol{X}_{ij}^P))/N_1^{P}.
\end{align*}
Therefore, by substituting $Y=f_0(\boldsymbol{X})+f_i(\boldsymbol{X})+\varepsilon$ into \eqref{T4(II)sta}, We can further express (II)
\begin{align}
    (II)& \leq\mathbb{E}[L^{\varepsilon}_{\hat{r}}(\hat{f}_{0,\hat{r}})-2\widehat{L}_{\hat{r}}^{\varepsilon}(\hat{f}_{0,\hat{r}})]+\mathbb{E}[L_{\hat{r}}^{g}(\hat{f}_{0,\hat{r}})-2\widehat{L}_{\hat{r}}^{g}(\hat{f}_{0,\hat{r}})]+2\Gamma\inf_{f\in \mathcal{F}}\|f-f_0\|_{L_2(P_{\mathbf{X}})}^2\nonumber\\
    & :=(A)+(B)+2\Gamma\inf_{f\in \mathcal{F}}\|f-f_0\|_{L_2(P_{\mathbf{X}})}^2.
    \label{T4(II)AB}
\end{align}
To bound (A), we break this part into three terms and analyze them separately. First, we truncate the independent noise terms, with the truncation function defined as follows: $\varepsilon_{ij}^\tau=\max(\min(\varepsilon_{ij},\tau),-\tau),\tau>0$. Let $Y^{\varepsilon^{\tau},P}_{ij}=f_0(\boldsymbol{X}^P_{ij})+\varepsilon^{\tau}_{ij}$. Naturally, we introduce the definition of $L^{\varepsilon}_{\hat{r}}(f)$ and $\hat{L}_{\hat{r}}^{\varepsilon}(f)$ in its truncated form as followed
\begin{align*}
 & L^{\varepsilon,\tau}_{\hat{r}}(f)=\mathbb{E}[\hat{r}(\boldsymbol{X}_{ij}^P)((Y^{\varepsilon^\tau,P}_{ij}-f(\boldsymbol{X}^P_{ij}))^{2}-(Y^{\varepsilon^\tau,P}_{ij}-f_0(\boldsymbol{X}^P_{ij}))^{2})], \\
 & \hat{L}_{\hat{r}}^{\varepsilon,\tau}(f)=\sum_{i=1}^{n_1^{P}}\sum_{j=1}^{m^{P}_{i}}[\hat{r}(\boldsymbol{X}_{i j}^P)((Y^{\varepsilon^\tau,P}_{i j}-f(\boldsymbol{X}_{ij}^P))^{2}-(Y^{\varepsilon^\tau,P}_{i j}-f_0(\boldsymbol{X}_{ij}^P))^{2})]/N_1^{P}.
\end{align*}
Then we have
\begin{align}
    (A)& =\mathbb{E}[L^{\varepsilon}_{\hat{r}}(\hat{f}_{0,\hat{r}})-L^{\varepsilon,\tau}_{\hat{r}}(\hat{f}_{0,\hat{r}})]+\mathbb{E}[L_{\hat{r}}^{\varepsilon,\tau}(\hat{f}_{0,\hat{r}})-2\widehat{L}_{\hat{r}}^{\varepsilon,\tau}(\hat{f}_{0,\hat{r}})]+\mathbb{E}[2\widehat{L}_{\hat{r}}^{\varepsilon,\tau}(\hat{f}_{0,\hat{r}})-2\widehat{L}_{\hat{r}}^{\varepsilon}(\hat{f}_{0,\hat{r}})]\nonumber\\
    & :=A_1(\hat{f}_{0,\hat{r}})+A_2(\hat{f}_{0,\hat{r}})+A_3(\hat{f}_{0,\hat{r}}).
    \label{T4(A)A1A2A3}
\end{align}
We first control $A_1(\hat{f}_{0,\hat{r}})$ and $A_3(\hat{f}_{0,\hat{r}})$.
\begin{align*}
|A_3(\hat{f}_{0,\hat{r}})| & =|\mathbb{E}[2\sum_{i=1}^{n_1^{P}}\sum_{j=1}^{m^{P}_{i}}(2\varepsilon_{ij}^{\tau}-2\varepsilon_{ij})\hat{r}(\boldsymbol{X}_{ij}^P)(f_0(\boldsymbol{X}_{ij}^P)-\hat{f}_{0,\hat{r}}(\boldsymbol{X}_{ij}^P))/N_1^{P}]| \\
 & \leq8B_1\Gamma\sum_{i=1}^{n_1^{P}}\sum_{j=1}^{m^{P}_{i}}\mathbb{E}[|\varepsilon_{ij}^\tau-\varepsilon_{ij}|]/N_1^{P} \\
 & =8B_1\Gamma\sum_{i=1}^{n_1^{P}}\sum_{j=1}^{m^{P}_{i}}\mathbb{E}[(|\varepsilon_{ij}|-\tau)I(|\varepsilon_{ij}|>\tau)]/N_1^{P} \\
 & \leq8B_1\Gamma\sum_{i=1}^{n_1^{P}}\sum_{j=1}^{m^{P}_{i}}\mathbb{E}[|\varepsilon_{ij}|I(|\varepsilon_{ij}|>\tau)]/N_1^{P}.
\end{align*}
Applying the facts
\begin{center}
    $|\varepsilon_{ij}|\leq2B_3\exp\{|\varepsilon_{ij}|/(2B_3)\}\quad\mathrm{and}\quad I(|\varepsilon_{ij}|>\tau)\leq\exp\{(|\varepsilon_{ij}|-\tau)/(2B_3)\}$,
\end{center}
by Assumption \ref{assumption fiepsilon}, we have
\begin{center}
    $|A_3(\hat{f}_{0,\hat{r}})|\leq16\Gamma B_1B_3\mathbb{E}[\exp\{|\varepsilon_{ij}|/B_3-\tau/(2B_3)\}]\leq16\Gamma B_1B_3\exp\{-\tau/(2B_3)\}$.
\end{center}
Similarly, $|A_1(\hat{f}_{0,\hat{r}})|\leq8\Gamma B_1B_3\exp\{-\tau/(2B_3)\}$, and setting $\tau=2B_3\log N_1^P$, we have
\begin{equation}
    \label{T4A1A3result}
    A_1(\hat{f}_{0,\hat{r}})+A_3(\hat{f}_{0,\hat{r}})\leq24\Gamma B_1B_3/N_1^P.
\end{equation}

Then we use the the Bernstein concentration inequality to control $A_2(\hat{f}_{0,\hat{r}})$, which is similar to the method in the proof of Theorem \ref{Theorem rstochastic}. Define $\tilde{\ell}(f;\boldsymbol{W}_{ij}):=\hat{r}(\boldsymbol{X}_{i j}^P)((Y_{ij}^{\varepsilon^\tau,P}-f(\boldsymbol{X}_{ij}^P))^2-(Y_{ij}^{\varepsilon^\tau,P}-f_0(\boldsymbol{X}_{ij}^P))^2)$ with $\boldsymbol{W}_{i j}:=(\boldsymbol{X}_{ij}^P,\varepsilon_{ij})$. Let $\mathcal{W}^{\prime}:=\{(\boldsymbol{X}_{ij}^{P^{\prime}},\varepsilon_{ij}^{\prime}):1\leq i\leq n_1^P, 1\leq j\leq m_i^P\}$ and be independent copies of $\mathcal{W}:=\{(\boldsymbol{X}_{ij}^{P},\varepsilon_{ij}):1\leq i\leq n_1^P, 1\leq j\leq m_i^P\}$. Define $G(f;\boldsymbol{W}_{ij}):=\mathbb{E}_{\mathcal{W}^{\prime}}[\tilde{\ell}(f;\boldsymbol{W}_{ij}^{\prime})-2\tilde{\ell}(f;\boldsymbol{W}_{ij})]$. Then we have
\begin{align*}
A_{2}(\hat{f}_{0,\hat{r}}) & =\mathbb{E}[\sum_{i=1}^{n_1^{P}}\sum_{j=1}^{m^{P}_{i}}(\mathbb{E}_{\mathcal{W}^{\prime}}[\tilde{\ell}(\hat{f}_{0,\hat{r}};\boldsymbol{W}_{ij}^{\prime})]-2\tilde{\ell}(\hat{f}_{0,\hat{r}};\boldsymbol{W}_{ij}))/N_1^{P}] \\
 & =\mathbb{E}[\sum_{i=1}^{n_1^{P}}\sum_{j=1}^{m^{P}_{i}}G(\hat{f}_{0,\hat{r}};\boldsymbol{W}_{ij})/N_1^{P}].
\end{align*}
Let $\mathcal{N}_{\mathcal{F}}:=\mathcal{N}(\epsilon,\mathcal{F},\|\cdot\|_{L_{2}(P_{\boldsymbol{X}})})$ be the $\epsilon$-covering number of $\mathcal{F}$ with respect to $\|\cdot\|_{L_2(P_X)}$. Accordingly, there exists an $\epsilon$-cover $\mathcal{C}_{\mathcal{F}}=\{f_1,\ldots,f_{\mathcal{N}_{\mathcal{F}}}\}\subset\mathcal{F}$. Then, for any $f\in\mathcal{F}$, there exists a $\tilde{f}\in\mathcal{C}_{\mathcal{F}}$ such that 
\begin{align*}
 & \mathbb{E}|\tilde{\ell}(f;\boldsymbol{W}_{ij})-\tilde{\ell}(\tilde{f};\boldsymbol{W}_{ij})|\leq2\Gamma(2B_1+\tau)\epsilon, \\
 & \mathbb{E}[G(f;\boldsymbol{W}_{ij})-G(\tilde{f};\boldsymbol{W}_{ij})]\leq6\Gamma(2B_1+\tau)\epsilon.
\end{align*}
For each $\tilde{f}\in \mathcal{C}_{\mathcal{F}}$, we have
\begin{align*}
    &~~~~~~~~\mathbb{P}[\sum_{i=1}^{n_1^{P}}\sum_{j=1}^{m^{P}_{i}}G(\tilde{f};\boldsymbol{W}_{ij})/N_1^{P}>t]=\mathbb{P}[\sum_{i=1}^{n_1^{P}}\sum_{j=1}^{m^{P}_{i}}(\mathbb{E}_{\mathcal{W}^{\prime}}[\tilde{\ell}(\tilde{f};\boldsymbol{W}_{ij}^{\prime})]-2\tilde{\ell}(\tilde{f};\boldsymbol{W}_{ij}))/N_1^{P}>t]\\
    &=\mathbb{P}[\mathbb{E}_\mathcal{W}[\sum_{i=1}^{n_1^{P}}\sum_{j=1}^{m^{P}_{i}}\tilde{\ell}(\tilde{f};\boldsymbol{W}_{ij})/N_1^{P}]-\sum_{i=1}^{n_1^{P}}\sum_{j=1}^{m^{P}_{i}}\tilde{\ell}(\tilde{f};\boldsymbol{W}_{ij})/N_1^{P}>t/2+\mathbb{E}_{\mathcal{W}}[\sum_{i=1}^{n_1^{P}}\sum_{j=1}^{m^{P}_{i}}\tilde{\ell}(\tilde{f};\boldsymbol{W}_{ij})/N_1^{P}]].
\end{align*}
It is easy to show that $|\tilde{\ell}(\tilde{f};\boldsymbol{W}_{ij})|\leq4\Gamma B_1(\tau+2B_1)$, and $|\tilde{\ell}(\tilde{f};\boldsymbol{W}_{ij})-\mathbb{E}\tilde{\ell}(\tilde{f};\boldsymbol{W}_{ij})|\leq8\Gamma B_1(\tau+2B_1):=b_1$. Denote by $\sigma_1^2:=Var(\tilde{\ell}(\tilde{f},\boldsymbol{W}_{ij}))$, then we have
\begin{align*}
\sigma^{2}_1 &\leq\mathbb{E}[\tilde{\ell}(\tilde{f};\boldsymbol{W}_{ij})^{2}]=\mathbb{E}[\hat{r}(\boldsymbol{X}_{i j}^P)((f_{0}(\boldsymbol{X}_{ij}^P)+\varepsilon^{\tau}_{ij}-\tilde{f}(\boldsymbol{X}_{ij}^P))^{2}-(f_{0}(\boldsymbol{X}_{ij}^P)+\varepsilon^{\tau}_{ij}-f_0(\boldsymbol{X}_{ij}^P))^{2})]^{2} \\
 & =\mathbb{E}[\hat{r}(\boldsymbol{X}_{i j}^P)(\tilde{f}(\boldsymbol{X}_{ij}^P)-f_0(\boldsymbol{X}_{ij}^P))^2+2\hat{r}(\boldsymbol{X}_{i j}^P)\varepsilon^\tau_{ij}(f_0(\boldsymbol{X}_{ij}^P)-\tilde{f}(\boldsymbol{X}_{ij}^P))]^2 \\
 & 
\leq2\Gamma\mathbb{E}[\hat{r}(\boldsymbol{X}_{i j}^P)(\tilde{f}(\boldsymbol{X}_{ij}^P)-f_0(\boldsymbol{X}_{ij}^P))^4]+2\mathbb{E}[(2\hat{r}(\boldsymbol{X}_{i j}^P)\varepsilon^\tau_{ij}(f_0(\boldsymbol{X}_{ij}^P)-\tilde{f}(\boldsymbol{X}_{ij}^P)))^2]\\
 & \leq8\Gamma B_1^2\mathbb{E}[\hat{r}(\boldsymbol{X}_{i j}^P)(\tilde{f}(\boldsymbol{X}_{ij}^P)-f_0(\boldsymbol{X}_{ij}^P))^2]+8\Gamma\tau^2\mathbb{E}[\hat{r}(\boldsymbol{X}_{i j}^P)(\tilde{f}(\boldsymbol{X}_{ij}^P)-f_0(\boldsymbol{X}_{ij}^P))^2] \\
 & =8\Gamma(B_{1}^{2}+\tau^{2})\mathbb{E}[\tilde{\ell}(\tilde{f};\boldsymbol{W}_{ij})].
\end{align*}
Thus, we have $\mathbb{E}[\tilde{\ell}(\tilde{f};\boldsymbol{W}_{ij})]\geq\sigma^2_1/(8\Gamma(B_1^2+\tau^2))$. Let $u_1:=t/2+\sigma^2_1/(16\Gamma(B_1^2+\tau^2))$, it is easy to show that $\sigma^2/u_1\leq16\Gamma(B_1^2+\tau^2)$ and $u_1\geq t/2$. By Bernstein concentration inequality, we have
\begin{align*}
&\mathbb{P}[\sum_{i=1}^{n_1^{P}}\sum_{j=1}^{m^{P}_{i}}G(\tilde{f};\boldsymbol{W}_{ij})/N_1^{P}>t]\\
& ~
\leq\mathbb{P}[\mathbb{E}[\sum_{i=1}^{n_1^{P}}\sum_{j=1}^{m^{P}_{i}}\tilde{\ell}(\tilde{f};\boldsymbol{W}_{ij})/N_1^{P}]-\sum_{i=1}^{n_1^{P}}\sum_{j=1}^{m^{P}_{i}}\tilde{\ell}(\tilde{f};\boldsymbol{W}_{ij})/N_1^{P}>t/2+\sigma^2/(16\Gamma^2(B_1^2+\tau^2))] \\
 &~ =\mathbb{P}[\mathbb{E}[\sum_{i=1}^{n_1^{P}}\sum_{j=1}^{m^{P}_{i}}\tilde{\ell}(\tilde{f};\boldsymbol{W}_{ij})/N_1^{P}]-\sum_{i=1}^{n_1^{P}}\sum_{j=1}^{m^{P}_{i}}\tilde{\ell}(\tilde{f};\boldsymbol{W}_{ij})/N_1^{P}>u_1] \\
 & ~\leq\exp(-N_1^{P}u_1^2/2(\sigma_1^2+b_1u_1))\\
 & ~\leq\exp(-N_1^{P}u_1/2(b_1+16\Gamma(B_1^2+\tau^2))) \\
 & ~\leq\exp(-N_1^{P}t/32\Gamma(4B_1^2+B_1\tau+2\tau^2)).
\end{align*}
Setting $a=6\Gamma(2B_1+\tau)\epsilon+32\Gamma(4B_1^2+B_1\tau+2\tau^2)\log\mathcal{N}_{\mathcal{F}}/N_1^{P}$ with $\epsilon=B_1/N_1^{P}$ and $\tau=2B_3\log N_1^{P}$, we have
\begin{align}
&A_{2}(\hat{f}_{0,\hat{r}}) 
=\mathbb{E}[\sum_{i=1}^{n_1^{P}}\sum_{j=1}^{m^{P}_{i}}G(\hat{f}_{0,\hat{r}};\boldsymbol{W}_{ij})/N_1^{P}]\nonumber\\
 & \leq\mathbb{E}[\max_{f\in\mathcal{F}}\sum_{i=1}^{n_1^{P}}\sum_{j=1}^{m^{P}_{i}}G(f;\boldsymbol{W}_{ij})/N_1^{P}]  \leq\mathbb{E}[\max_{\tilde{f}\in\mathcal{C}_{\mathcal{F}}}\sum_{i=1}^{n_1^{P}}\sum_{j=1}^{m^{P}_{i}}G(\tilde{f};\boldsymbol{W}_{ij})/N_1^{P}]+6\Gamma(2B_{1}+\tau)\epsilon \nonumber\\
 & \leq a+\int_a^\infty\mathcal{N}_{\mathcal{F}}\exp\{-N_1^{P}(t-6\Gamma(2B_1+\tau)\epsilon)/(32\Gamma(4B_1^2+B_1\tau+2\tau^2))\}\mathrm{d}t \nonumber\\
 & \leq a+32\Gamma\mathcal{N}_{\mathcal{F}}(4B_1^2+B_1\tau+2\tau^2)\exp\{-N_1^{P}(a-6\Gamma(2B_1+\tau)\epsilon)/(32\Gamma(4B_1^2+B_1\tau+2\tau^2))\}/N_1^{P} \nonumber\\
 & =6\Gamma(2B_1+\tau)\epsilon+32\Gamma(4B_1^2+B_1\tau+2\tau^2)(\log\mathcal{N}_{\mathcal{F}}+1) /N_1^{P}\nonumber\\
 & \leq256\Gamma(B_3\log N_1^{P}+B_1)^2(\log\mathcal{N}_{\mathcal{F}}+2)/N_1^{P}.  
 \label{T4A2result}
\end{align}
Combining the results of \eqref{T4(A)A1A2A3}, \eqref{T4A1A3result} and \eqref{T4A2result}, we obtain
\begin{align}
    (A)\leq24\Gamma B_1B_3/{N_1^P}+256\Gamma(B_3\log N_1^P+B_1)^2(\log\mathcal{N}_{\mathcal{F}}+2)/{N_1^P}.
    \label{T4(A)result}
\end{align}

Next, we aim to bound $(B)$ in \eqref{T4(II)AB}. We can obtain 
\begin{align*}
    (B)=\mathbb{E}[L_{\hat{r}}^g(\hat{f}_{0,\hat{r}})-2\widehat{L}_{\hat{r}}^g(\hat{f}_{0,\hat{r}})]=\mathbb{E}[2\sum_{i=1}^{n_1^P}\sum_{j=1}^{m_i^P}\hat{r}(\boldsymbol{X}_{i j}^P)f_i(\boldsymbol{X}_{ij}^P)(\hat{f}_{0,\hat{r}}(\boldsymbol{X}_{ij}^P)-f_0(\boldsymbol{X}_{ij}^P))/N_1^P],
\end{align*}
according to the proof of Theorem 2.1 given in \citet{Yan02102025}, it follows that
\begin{align}
(B) & \leq\mathbb{E}[\|\hat{f}_{0,\hat{r}}-f_0\|_{L_2(Q_{\boldsymbol{X}})}^2]/16+c_{1}\Gamma(B_{1}^{2}+B_{2}^{2})/{n_1^P}+c_{2}\Gamma(B_{1}^{2}+B_{2}^{2}\operatorname{log}^{2}n_1^P)/{N_1^P} \nonumber\\
 & +c_3\Gamma B_2^2\log^2n_1^P(\log\mathcal{N}_{\mathcal{F}}+\log N_1^P)/{N_1^P},
 \label{T4(B)result}
\end{align} 
where $c_1$, $c_2$ and $c_3$ are constants. Combining the results of \eqref{T4(I)(II)}, \eqref{T4(I)result}, \eqref{T4(II)AB}, \eqref{T4(A)result} and \eqref{T4(B)result}, and then using the results of Theorems \ref{Theorem approx}, \ref{Theorem 3rb}, \eqref{logN_VC} and \eqref{VC_LW}, we have
\begin{align*}
    \mathbb{E}[\|\hat{f}_{0,\hat{r}}-f_0\|_{L_2(Q_{\boldsymbol{X}})}^2]& \lesssim 8(B_1^2+B_2^2+B_3^2)\mathbb{E}[\|r-\hat{r}\|_{L_2(P_{\boldsymbol{X}})}^2]/{ M}\\
    &+\Gamma/{n_1^P}+\Gamma\mathcal{DS}\log\mathcal{S}(\log N_1^P)^3/{N_1^P}.
\end{align*}
 Finally, setting the depth $\mathcal{D}=O(\log N_1^P)$, the weights $\mathcal{B}=O((N_1^P)^{d/(d+2\zeta)})$, and the size $\mathcal{S}=O((N_1^P)^{d/(2\zeta+d)}(\log N_1^P)^{-5d/(2\zeta+d)})$, we can obtain 
\begin{align*}
    \mathbb{E}[\|\hat{f}_{0,\hat{r}}-f_0\|_{L_2(Q_{\boldsymbol{X}})}^2] \lesssim & \Gamma((n_1^P)^{-1}+(N_1^P)^{-2\zeta/(2\zeta+d)}(\log N_1^P)^{10\zeta/(2\zeta+d)})
    \\ 
    &+\Gamma(\Gamma+1)(N_r)^{-2\alpha/(d+2\alpha)}(\log N_r)^{5}/{M}.
\end{align*}
Moreover, if $N_r\gtrsim((\Gamma+1)/{M})^{(d+2\alpha)/{2\alpha}}(N_1^P)^{\zeta(d+2\alpha)/(\alpha(d+2\zeta))}$, then we have
\begin{align*}
    \mathbb{E}[\|\hat{f}_{0,\hat{r}}-f_0\|_{L_2(Q_{\boldsymbol{X}})}^2] \lesssim \Gamma((n_1^P)^{-1}+(N_1^P)^{-2\zeta/(2\zeta+d)}(\log N_1^P)^{10\zeta/(2\zeta+d)}).
\end{align*}
Thus we complete the proof.
\end{proof}

\subsubsection{Proof of Theorem \ref{Theorem Ef^_0_r^_zeta}}
\begin{proof}
First, we similarly decompose the error term. Through straightforward calculations, we have
\begin{align}
    &\mathbb{E}[\|\hat{f}_{0,\hat{r}_{\xi}}-f_0\|_{L_2(Q_{\boldsymbol{X}})}^2]\nonumber
     =\mathbb{E}[R(\hat{f}_{0,\hat{r}_{\xi}})-R(f_0)]\nonumber\\
     =&\mathbb{E}[\mathbb{E}_{(\boldsymbol{X},Y)\thicksim P_{\boldsymbol{X},Y}}[(r(\boldsymbol{X})-\hat{r}_{\xi}(\boldsymbol{X}))((Y-\hat{f}_{0,\hat{r}_{\xi}}(\boldsymbol{X}))^2-(Y-f_0(\boldsymbol{X}))^2)]]\nonumber\\
    & +\mathbb{E}[\mathbb{E}_{(\boldsymbol{X},Y)\thicksim P_{\boldsymbol{X},Y}}[\hat{r}_{\xi}(\boldsymbol{X})((Y-\hat{f}_{0,\hat{r}_{\xi}}(\boldsymbol{X}))^2-(Y-f_0(\boldsymbol{X}))^2)]]\nonumber\\
     :=&(I)+(II).
     \label{T6(I)(II)}
\end{align}
According to \eqref{T4(I)result}, we similarly obtain
\begin{align}
    (I)\leq8(B_1^2+B_2^2+B_3^2)\mathbb{E}[\|r-\hat{r}_{\xi}\|_{L_2(P_{\boldsymbol{X}})}^2]/{ M}+\mathbb{E}[\|\hat{f}_{0,\hat{r}_{\xi}}-f_0\|_{L_2(Q_{\boldsymbol{X}})}^2]/2.
    \label{T6(I)result}
\end{align}
Then, using the same proof approach as in Theorem \ref{Theorem Ef^_0_r^}, we obtain 
\begin{align}
    (II)\lesssim &\xi( (B_2\log n_1^P)^2+(B_3\log N_1^P+B_1)^2)(\log\mathcal{N}_{\mathcal{F}}+\log N_1^P)/{N_1^P}\nonumber\\
    &+\xi\inf_{f\in \mathcal{F}}\|f-f_0\|_{L_2(P_{\mathbf{X}})}^2
    +\xi(B_{1}^{2}+B_{2}^{2})/{n_1^P}+\mathbb{E}[\|\hat{f}_{0,\hat{r}_{\xi}}-f_0\|_{L_2(Q_{\boldsymbol{X}})}^2]/16.
    \label{T6(II)result}
\end{align}
By the results of Theorems \ref{Theorem approx}, \ref{Theorem r_xi_hat}, \eqref{logN_VC} and \eqref{VC_LW}, combining the results of \eqref{T6(I)(II)}, \eqref{T6(I)result} and \eqref{T6(II)result}, we obtain
\begin{align*}
    \mathbb{E}[\|\hat{f}_{0,\hat{r}_{\xi}}-f_0\|_{L_2(Q_{\boldsymbol{X}})}^2] &\lesssim \xi((n_1^P)^{-1}+\mathcal{DS}\log\mathcal{S}(\log N_1^P)^3/N_1^P)\\
    & +\kappa_{\delta}^{2/(2+\delta)}(N_r)^{-2\alpha/(d+(2+4/\delta)\alpha)}(\log N_r)^5/M.
\end{align*}
Setting $\xi=O((1/n_1^P+\mathcal{DS}\log\mathcal{S}(\log N_1^P)^3/N_1^P)^{-1/(\delta+2)})$, the size $\mathcal{S}=O((N_1^P)^{d/(d+(2+2/(\delta+1))\zeta)}\\(\log N_1^P)^{-5d/(2\zeta+d)})$, the depth $\mathcal{D}=O(\log N_1^P)$ and the weights $\mathcal{B}=O((N_1^P)^{d/(d+(2+2/(\delta+1))\zeta)})$, we have
\begin{align*}
\mathbb{E}[\|\hat{f}_{0,\hat{r}_{\xi}}-f_{0}\|_{L_{2}(Q_{X})}^{2}] & \lesssim (n_1^P)^{-(\delta+1)/(\delta+2)}+(N_1^P)^{-2\zeta/(d+(2+2/(\delta+1))\zeta)}(\log N_1^P)^{10\zeta/(2\zeta+d)} \\
 & +\kappa_{\delta}^{2/(2+\delta)}(N_r)^{-2\alpha/(d+(2+4/\delta)\alpha)}(\log N_r)^5/M.
\end{align*}
Moreover, if $N_r\gtrsim M^{-\eta_{1}}{\kappa_{\delta}}^{\eta_2}(N_1^P)^{\eta_3}$ with $\eta_1=\{d\delta+2(\delta+2)\alpha\}/\{2\alpha\delta\}$, $\eta_2=\{\delta d+2(\delta+2)\alpha\}/\{\alpha\delta(\delta+2)\}$ and $\eta_3=\{2\zeta\alpha(\delta+2)(\delta+1)+d\zeta\delta(\delta+1)\}/\{2\zeta\alpha\delta(\delta+2)+d\alpha\delta(\delta+1)\}$, then we have
\begin{center}
    $\mathbb{E}[\|\hat{f}_{0,\hat{r}_{\xi}}-f_{0}\|_{L_{2}(Q_{X})}^{2}]\lesssim (n_1^P)^{-(\delta+1)/(\delta+2)}+(N_1^P)^{-2\zeta/(d+(2+2/(\delta+1))\zeta)}(\log N_1^P)^{10\zeta/(2\zeta+d)}$.
\end{center} 
Thus we complete the proof.
\end{proof}
\subsection{Proofs for Target Regression Estimators with a Known Density ratio}\label{ProofKRE}
In this section, we give the proofs of Theorems \ref{Theorem fhatr} and \ref{Theorem fhatr_xi}, which establish the non-asymptotic error upper bounds for the target regression estimator with a known density ratio under bounded and unbounded density ratio settings, respectively. These proofs follow similar approaches to those in Theorems \ref{Theorem Ef^_0_r^} and \ref{Theorem Ef^_0_r^_zeta}.

\subsubsection{Proof of Theorem \ref{Theorem fhatr}}

\begin{proof}
First, we similarly decompose the error term. Recall the definition of $R(f)=\mathbb{E}_{(\boldsymbol{X},Y)\thicksim P_{\boldsymbol{X},Y}}[r(\boldsymbol{X})(Y-f(\boldsymbol{X}))^2]$. Through straightforward calculations, we have $\mathbb{E}[\|\hat{f}_{0,r}-f_0\|_{L_2(Q_{\boldsymbol{X}})}^2]=\mathbb{E}[R(\hat{f}_{0,r})-R(f_0)]$. For each $f(\cdot)\in \mathcal{F}$, we define
\begin{align*}
    & L_{r}(f)=\mathbb{E}_{\boldsymbol{X},Y\sim P_{\boldsymbol{X},Y}}[r(\boldsymbol{X})((Y-f(\boldsymbol{X}))^2-(Y-f_0(\boldsymbol{X}))^2)],\\
    & \widehat{L}_{r}(f)=\sum_{i=1}^{n^{P}}\sum_{j=1}^{m^{P}_{i}}
   r(\boldsymbol{X}_{ij}^{P})((Y_{ij}^{P} - f(\boldsymbol{X}_{ij}^{P}))^2/N^{P}-(Y_{ij}^{P} - f_0(\boldsymbol{X}_{ij}^{P}))^2).
\end{align*}
Following the same procedure as in the derivation of \eqref{T4(II)sta}, we can analogously obtain
\begin{align}
    \mathbb{E}[L_{r}(\hat{f}_{0,r})]\leq\mathbb{E}[L_{r}(\hat{f}_{0,r})-2\widehat{L}_{r}(\hat{f}_{0,r})]+2\Gamma\inf_{f\in \mathcal{F}}\|f-f_0\|_{L_2(P_{\mathbf{X}})}^2.
     \label{T7fj1}
\end{align}
We alse let $Y^{\varepsilon,P}_{ij}=f_0(\boldsymbol{X}_{ij}^P)+\varepsilon_{ij}$ and define 
\begin{align*}
 & L^{\varepsilon}_{r}(f)=\mathbb{E}[r(\boldsymbol{X}_{ij}^P)((Y^{\varepsilon,P}_{ij}-f(\boldsymbol{X}_{ij}^P))^{2}-(Y^{\varepsilon,P}_{ij}-f_0(\boldsymbol{X}_{ij}^P))^{2})],\\
 & \hat{L}_{r}^{\varepsilon}(f)=\sum_{i=1}^{n^{P}}\sum_{j=1}^{m^{P}_{i}}[r(\boldsymbol{X}_{i j}^P)((Y_{i j}^{\varepsilon,P}-f(\boldsymbol{X}_{ij}^P))^{2}-(Y_{i j}^{\varepsilon,P}-f_0(\boldsymbol{X}_{ij}^P))^{2})]/N^{P},\\
 & L^{g}_{r}(f)=\mathbb{E}_{\boldsymbol{X}\sim P_{\boldsymbol{X}}}[2r(\boldsymbol{X})f_{i}(\boldsymbol{X})(f_0(\boldsymbol{X})-f(\boldsymbol{X}))],\\
 & \widehat{L}_{r}^{g}(f)=2\sum_{i=1}^{n^{P}}\sum_{j=1}^{m^{P}_{i}}r(\boldsymbol{X}_{i j}^P)f_{i}(\boldsymbol{X}_{ij}^P)(f_0(\boldsymbol{X}_{ij}^P)-f(\boldsymbol{X}_{ij}^P))/N^{P}.
\end{align*}
Therefore, by substituting $Y=f_0(\boldsymbol{X})+f_i(\boldsymbol{X})+\varepsilon$ into the \eqref{T7fj1}, We can further express $\mathbb{E}[L_{r}(\hat{f}_{0,r})]$ as
\begin{align*}
    \mathbb{E}[L_{r}(\hat{f}_{0,r})]& \leq\mathbb{E}[L^{\varepsilon}_{r}(\hat{f}_{0,r})-2\widehat{L}_{r}^{\varepsilon}(\hat{f}_{0,r})]+\mathbb{E}[L_{r}^{g}(\hat{f}_{0,r})-2\widehat{L}_{r}^{g}(\hat{f}_{0,r})]+2\Gamma\inf_{f\in \mathcal{F}}\|f-f_0\|_{L_2(P_{\mathbf{X}})}^2\nonumber\\
    & :=(A)+(B)+2\Gamma\inf_{f\in \mathcal{F}}\|f-f_0\|_{L_2(P_{\mathbf{X}})}^2.
\end{align*}

To bound (A) and (B), we similarly truncate the independent noise terms, with the truncation function defined as follows: $\varepsilon_{ij}^\tau=\max(\min(\varepsilon_{ij},\tau),-\tau),\tau>0$. Let $Y^{\varepsilon^\tau,P}_{ij}=f_0(\boldsymbol{X}_{ij}^P)+\varepsilon^{\tau}_{ij}$. Naturally, we introduce the definition of $L^{\varepsilon}_{r}(f)$ and $\hat{L}_{r}^{\varepsilon}(f)$ in its truncated form as followed
\begin{align*}
 & L^{\varepsilon,\tau}_{r}(f)=\mathbb{E}[r(\boldsymbol{X}_{ij}^P)((Y^{\varepsilon^\tau,P}_{ij}-f(\boldsymbol{X}_{ij}^P))^{2}-(Y^{\varepsilon^\tau,P}_{ij}-f_0(\boldsymbol{X}_{ij}^P))^{2})], \\
 & \hat{L}_{r}^{\varepsilon,\tau}(f)=\sum_{i=1}^{n^{P}}\sum_{j=1}^{m^{P}_{i}}[r(\boldsymbol{X}_{i j}^P)((Y^{\varepsilon^\tau,P}_{i j}-f(\boldsymbol{X}_{ij}^P))^{2}-(Y^{\varepsilon^\tau,P}_{i j}-f_0(\boldsymbol{X}_{ij}^P))^{2})]/N^{P}.
\end{align*}
Then we have
\begin{align*}
    (A)& =\mathbb{E}[L^{\varepsilon}_{r}(\hat{f}_{0,r})-L^{\varepsilon,\tau}_{r}(\hat{f}_{0,r})]+\mathbb{E}[L_{r}^{\varepsilon,\tau}(\hat{f}_{0,r})-2\widehat{L}_{r}^{\varepsilon,\tau}(\hat{f}_{0,r})]+\mathbb{E}[2\widehat{L}_{r}^{\varepsilon,\tau}(\hat{f}_{0,r})-2\widehat{L}_{r}^{\varepsilon}(\hat{f}_{0,r})]\nonumber\\
    & :=A_1(\hat{f}_{0,r})+A_2(\hat{f}_{0,r})+A_3(\hat{f}_{0,r}).
\end{align*}
Similarly, Define $\tilde{\ell}_{r}(f;\boldsymbol{W}_{ij}):=r(\boldsymbol{X}_{i j}^P)((Y_{ij}^{\varepsilon^\tau}-f(\boldsymbol{X}_{ij}^P))^2-(Y_{ij}^{\varepsilon^\tau}-f_0(\boldsymbol{X}_{ij}^P))^2)$ with $\boldsymbol{W}_{i j}:=(\boldsymbol{X}_{ij}^P,\varepsilon_{ij})$. Let $\mathcal{W}^{\prime}:=\{(\boldsymbol{X}_{ij}^{P^{\prime}},\varepsilon_{ij}^{\prime}):1\leq i\leq n^P, 1\leq j\leq m_i^P\}$ and be independent copies of $\mathcal{W}:=\{(\boldsymbol{X}_{ij}^{P},\varepsilon_{ij}):1\leq i\leq n^P, 1\leq j\leq m_i^P\}$. Define $G_{r}(f;\boldsymbol{W}_{ij}):=\mathbb{E}_{\mathcal{W}^{\prime}}[\tilde{\ell}_{r}(f;\boldsymbol{W}_{ij}^{\prime})-2\tilde{\ell}_{r}(f;\boldsymbol{W}_{ij})]$. Therefore, employing analogous arguments in the proof of Theorem \ref{Theorem Ef^_0_r^} and setting the depth $\mathcal{D}=O(\log N^P)$, the weights $\mathcal{B}=O((N^P)^{d/(d+2\zeta)})$ and the size $\mathcal{S}=O((N^P)^{d/(2\zeta+d)}(\log N^P)^{-5d/(2\zeta+d)})$, we can obtain 
\begin{center}
    $\mathbb{E}[\|\hat{f}_{0,r}-f_0\|_{L_2(Q_{\boldsymbol{X}})}^2] \lesssim \Gamma((n^P)^{-1}+(N^P)^{-2\zeta/(2\zeta+d)}(\log N^P)^{10\zeta/(2\zeta+d)})$.
\end{center}
Thus we complete the proof. 
\end{proof}

\subsubsection{Proof of Theorem \ref{Theorem fhatr_xi}}

\begin{proof}
 Through straightforward calculations, we have
\begin{align}
    &\mathbb{E}[\|\hat{f}_{0,r_{\xi}}-f_0\|_{L_2(Q_{\boldsymbol{X}})}^2]\nonumber
     =\mathbb{E}[R(\hat{f}_{0,r_{\xi}})-R(f_0)]\nonumber\\
     =&\mathbb{E}[\mathbb{E}_{(\boldsymbol{X},Y)\thicksim P_{\boldsymbol{X},Y}}[(r(\boldsymbol{X})-r_{\xi}(\boldsymbol{X}))((Y-\hat{f}_{0,r_{\xi}}(\boldsymbol{X}))^2-(Y-f_0(\boldsymbol{X}))^2)]]\nonumber\\
    & +\mathbb{E}[\mathbb{E}_{(\boldsymbol{X},Y)\thicksim P_{\boldsymbol{X},Y}}[r_{\xi}(\boldsymbol{X})((Y-\hat{f}_{0,r_{\xi}}(\boldsymbol{X}))^2-(Y-f_0(\boldsymbol{X}))^2)]]\nonumber\\
     :=&(I)+(II).
    \label{T8(I)(II)}
\end{align}
We first control term (I). Using Markov inequality and Assumption \ref{assumption rub}, we have
\begin{align}
 & \mathbb{E}_{(\boldsymbol{X},Y)\thicksim P_{\boldsymbol{X},Y}}[(r(\boldsymbol{X})-r_{\xi}(\boldsymbol{X}))((Y-f(\boldsymbol{X}))^{2}-(Y-f_{0}(\boldsymbol{X}))^{2})] \nonumber\\
  =&\mathbb{E}_{(\boldsymbol{X},Y)\thicksim P_{\boldsymbol{X},Y}}[(r(\boldsymbol{X})-\xi)I(r(\boldsymbol{X})\geq\xi)((Y-f(\boldsymbol{X}))^{2}-(Y-f_{0}(\boldsymbol{X}))^{2})]
 \nonumber\\
  =&\mathbb{E}_{\boldsymbol{X}\sim P_{\boldsymbol{X}}}[(r(\boldsymbol{X})-\xi)I(r(\boldsymbol{X})\geq\xi)(f_0(\boldsymbol{X})-f(\boldsymbol{X}))^2] \nonumber\\
  \leq&4B_1^2\mathbb{E}_{\boldsymbol{X}\sim P_{\boldsymbol{X}}}[(r(\boldsymbol{X})-\xi)I(r(\boldsymbol{X})\geq\xi)]\nonumber\\
 \leq&4B_1^2\mathbb{E}_{\boldsymbol{X}\thicksim P_{\boldsymbol{X}}}[r(\boldsymbol{X})r^{\delta+1}(\boldsymbol{X})/{\xi^{\delta+1}}]+4B_1^2\xi\mathbb{E}_{\boldsymbol{X}\thicksim P_{\boldsymbol{X}}}[r^{2+\delta}(\boldsymbol{X})/{\xi^{2+\delta}}] \nonumber\\
 =&16B_1^2\kappa_{\delta}/{\xi^{\delta+1}}.
 \label{T8(I)result}
\end{align}
Then, employing analogous arguments to in the proof of Theorem \ref{Theorem Ef^_0_r^}, (II) can be bounded by
\begin{align}
    (II)\lesssim \xi/{n^P}+\xi\mathcal{DS}\log\mathcal{S}(\log N^P)^3/{N^P}.
    \label{T8(II)result}
\end{align}
Combining the result of \eqref{T8(I)(II)}, \eqref{T8(I)result} and \eqref{T8(II)result}, we can obtain 
\begin{center}
    $\mathbb{E}[\|\hat{f}_{0,r_{\xi}}-f_0\|_{L_2(Q_{\boldsymbol{X}})}^2]\lesssim\kappa_{\delta}/{\xi^{\delta+1}}+\xi(1/{n^P}+\mathcal{DS}\log\mathcal{S}(\log N^P)^3/{N^P})$.
\end{center}
Setting $\xi=(1/{(n^P\kappa_{\delta})}+\mathcal{DS}\log\mathcal{S}(\log N^P)^{3}/{(N^P\kappa_{\delta})})^{-1/{(2+\delta)}}$, depth $\mathcal{D}=O(\log N^P)$, weights $\mathcal{B}=O((N^P)^{d/{(d+(2/(\delta+1)+2)\zeta)}})$, size $\mathcal{S}=O((N^P)^{d/{(d+(2/(\delta+1)+2)\zeta)}}(\log N^P)^{-5d/{(2\zeta+d)}})$, we have
\begin{center}
    $\mathbb{E}\|\hat{f}_{0,r_{\xi}}-f_0\|_{L_2(Q_{\boldsymbol{X}})}^2\lesssim \kappa_{\delta}^{1/{(2+\delta)}}((n^P)^{-(\delta+1)/{(\delta+2)}}+(N^P)^{-2\zeta/{(d+(2/(\delta+1)+2)\zeta)}}(\log N^P)^{10\zeta/{(2\zeta+d)}})$.
\end{center}
Thus we complete the proof.
\end{proof}

\section{Theoretical Analysis of $\hat{f}_{0}^{P}(\cdot)$}\label{f0P}
In this part, we develop a theoretical analysis of the naive estimator $\hat{f}_0^{P}(\cdot)$ in \eqref{hatf0P}, including non-asymptotic error upper bounds under bounded density ratio and finite second moment assumptions, together with the corresponding proofs.

\setcounter{theorem}{0}
\renewcommand{\thetheorem}{C.\arabic{theorem}}
\begin{theorem}
    \label{Theorem fPbounded}
    Suppose that Assumptions \ref{assumption rsupp}, \ref{assumption rb}, \ref{assumption f0B1} and \ref{assumption fiepsilon} are satisfied, and let $\mathcal{F}$ be the function class in \eqref{FNN_d_1} bounded by $B_1\geq1$ with the depth $\mathcal{D}=O(\log N^P)$, the weights $\mathcal{B}=O((N^P)^{d/(d+2\zeta)})$ and the size $\mathcal{S}=O((N^P)^{d/(2\zeta+d)}(\log N^P)^{-5d/(2\zeta+d)})$. Then, the mean squared error of the estimator $\hat{f}_{0}^{P}(\cdot)$ in \eqref{hatf0P} satisfies
    \begin{center}
        $\mathbb{E}[\|\hat{f}_{0}^{P}-f_0\|_{L_2(Q_{\boldsymbol{X}})}^2] \lesssim \Gamma((n^P)^{-1}+(N^P)^{-2\zeta/(2\zeta+d)}(\log N^P)^{10\zeta/(2\zeta+d)})$.
    \end{center}
\end{theorem}
\begin{theorem}
    \label{Theorem fPunbounded}
    Suppose that Assumptions \ref{assumption rsupp}, \ref{assumption f0B1}, \ref{assumption fiepsilon} and $\kappa_{0}:=\mathbb{E}_{\boldsymbol{X}\sim P_{\boldsymbol{X}}}[r^{2}(\boldsymbol{X})]<\infty$ are satisfied, and let $\mathcal{F}$ be the function class in \eqref{FNN_d_1} bounded by $B_1\geq1$ with the depth $\mathcal{D}=O(\log N^P)$, the weights $\mathcal{B}=O((N^P)^{d/(d+2\zeta)})$ and the size $\mathcal{S}=O((N^P)^{d/(2\zeta+d)}(\log N^P)^{-5d/(2\zeta+d)})$. Then, the mean squared error of the estimator $\hat{f}_{0}^{P}(\cdot)$ in \eqref{hatf0P} satisfies
    \begin{center}
        $\mathbb{E}[\|\hat{f}_{0}^{P}-f_0\|_{L_2(Q_{\boldsymbol{X}})}^2] \lesssim \kappa_{0}^{1/2}((n^P)^{-1}+(N^P)^{-2\zeta/(2\zeta+d)}(\log N^P)^{10\zeta/(2\zeta+d)})^{1/2}$.
    \end{center}
\end{theorem}
Theorem \ref{Theorem fPbounded} establishes that, under a bounded density ratio, the naive estimator $\hat{f}_{0}^{P}(\cdot)$ attains the minimax optimal rate of $O((n^P)^{-1}+(N^P)^{-2\zeta/(2\zeta+d)})$ up to logarithmic factors. This indicates that mild covariate shift does not affect statistical optimality. In contrast to Theorem \ref{Theorem fPbounded}, Theorem \ref{Theorem fPunbounded} shows that under the finite second moment assumption, the estimator achieves the sub-optimal convergence rate of $O((n^P)^{-1/2}+(N^P)^{-\zeta/(2\zeta+d)})$ up to the constants and the logarithmic factors. We then give the proofs of the above theorems, following a similar approach to that in Theorem \ref{Theorem fhatr}. 

\subsection{Proof of Theorem \ref{Theorem fPbounded}}

\begin{proof}
First, we decompose the error term. We define $R_{P}(f)=\mathbb{E}_{(\boldsymbol{X},Y)\thicksim P_{\boldsymbol{X},Y}}[(Y-f(\boldsymbol{X}))^2]$. Through similarly straightforward calculations, we have $\mathbb{E}[\|\hat{f}_{0}^{P}-f_0\|_{L_2(P_{\boldsymbol{X}})}^2]=\mathbb{E}[R_{P}(\hat{f}_{0}^{P})-R_{P}(f_0)]$. For each $f(\cdot)\in \mathcal{F}$, we similar define
\begin{align*}
    & L_{P}(f)=\mathbb{E}_{\boldsymbol{X},Y\sim P_{\boldsymbol{X},Y}}[(Y-f(\boldsymbol{X}))^2-(Y-f_0(\boldsymbol{X}))^2],\\
    & \widehat{L}_{P}(f)=\sum_{i=1}^{n^{P}}\sum_{j=1}^{m^{P}_{i}}
[(Y_{ij}^{P} - f(\boldsymbol{X}_{ij}^{P}))^2-(Y_{ij}^{P} - f_0(\boldsymbol{X}_{ij}^{P}))^2]/N^{P}.
\end{align*}
Following the same procedure as in the derivation of Theorem \ref{Theorem fhatr}, we can analogously obtain
\begin{align}
    \mathbb{E}[\|\hat{f}_{0}^{P}-f_0\|_{L_2(P_{\boldsymbol{X}})}^2] \lesssim (n^P)^{-1}+(N^P)^{-2\zeta/(2\zeta+d)}(\log N^P)^{10\zeta/(2\zeta+d)}.
    \label{T9P}
\end{align}
Then by Assumption \ref{assumption rb} that $\Gamma:=\sup_{\boldsymbol{X}\in\mathcal{X}}r(\boldsymbol{X})<\infty$, there holds
\begin{align*}
    \mathbb{E}[\|\hat{f}_{0}^{P}-f_0\|_{L_2(Q_{\boldsymbol{X}})}^2]&\leq\Gamma\mathbb{E}[\|\hat{f}_{0}^{P}-f_0\|_{L_2(P_{\boldsymbol{X}})}^2]\\
    &\lesssim \Gamma((n^P)^{-1}+(N^P)^{-2\zeta/(2\zeta+d)}(\log N^P)^{10\zeta/(2\zeta+d)}).
\end{align*}
Thus we complete the proof.
\end{proof}

\subsection{Proof of Theorem \ref{Theorem fPunbounded}}

\begin{proof}
By the definition of $r(\boldsymbol{X})$, for each $f(\cdot)\in\mathcal{F}$, the Cauchy-Schwarz inequality yields
\begin{align*}
\|f-f_{0}\|_{L_{2}(Q_{X})}^{2} & =\mathbb{E}_{\boldsymbol{X}\sim P_{\boldsymbol{X}}}[ r(\boldsymbol{X})(f(\boldsymbol{X})-f_{0}(\boldsymbol{X}))^{2}] \\
 & \leq(\mathbb{E}_{\boldsymbol{X}\sim P_{\boldsymbol{X}}}[r^{2}(\boldsymbol{X})])^{1/2}(\mathbb{E}_{\boldsymbol{X}\sim P_{\boldsymbol{X}}}[(f(\boldsymbol{X})-f_{0}(\boldsymbol{X}))^{4}])^{1/2} \\
 & =\kappa_{0}^{1/2}(\mathbb{E}_{\boldsymbol{X}\sim P_{\boldsymbol{X}}}[(f(\boldsymbol{X})-f_{0}(\boldsymbol{X}))^{4}])^{1/2} \\
 & \leq(4B_{1}^{2}\kappa_{0})^{1/2}(\|f-f_{0}\|_{L_{2}(P_{X})}^{2})^{1/2}.
\end{align*}
Since $\hat{f}_0^P(\cdot)\in\mathcal{F}$, the above inequality can be applied to $f=\hat{f}_0^P$. Therefore, together with the result of \eqref{T9P}, we obtain
\begin{align*}
\mathbb{E}[\|\hat{f}_{0}^{P}-f_0\|_{L_2(Q_{\boldsymbol{X}})}^2] & \leq(4B_{1}^{2}\kappa_{0})^{1/2}(\mathbb{E}[\|\hat{f}_{0}^{P}-f_0\|_{L_2(P_{\boldsymbol{X}})}^2])^{1/2} \\
 & \lesssim \kappa_{0}^{1/2}((n^P)^{-1}+(N^P)^{-2\zeta/(2\zeta+d)}(\log N^P)^{10\zeta/(2\zeta+d)})^{1/2}.
\end{align*}
Thus we complete the proof.
\end{proof}

\end{document}